\documentclass[11pt,letterpaper,oneside]{article}

\usepackage[T1]{fontenc}
\usepackage[utf8]{inputenc}
\usepackage{lmodern}
\usepackage{microtype}

\usepackage[letterpaper,margin=1in]{geometry}

\usepackage{amsmath,amssymb,amsfonts,amsthm,mathtools,mathrsfs,bm}
\usepackage{array,booktabs,multirow,makecell}
\usepackage{enumitem}
\usepackage{graphicx}
\usepackage{tikz}
\usepackage{algorithm}
\usepackage{algorithmic}

\usepackage{nameref}

\usepackage[numbers,sort&compress]{natbib}

\usepackage[colorlinks=true,
  linkcolor=blue,
  citecolor=blue,
  urlcolor=blue,
  bookmarksdepth=2]{hyperref}
\usepackage[nameinlink,capitalise]{cleveref}

\usepackage{comment}
\usepackage{todonotes}

\usepackage{graphicx}
\usepackage{tikz}
\usetikzlibrary{cd,arrows,arrows.meta,positioning,shapes,fit,calc}
\tikzset{>=stealth}
\tikzcdset{arrow style=tikz}
\tikzset{link/.style={column sep=1.8cm,row sep=0.16cm}}

\AtBeginDocument{%
	\def\MR#1{}
}
\usepackage{amsmath}
\usepackage{mathdots}
\usepackage{amssymb}
\usepackage{amsthm}
\usepackage{here}
\usepackage{amscd} 
\usepackage{mathrsfs}
\usepackage{mathtools}
\usepackage{scalefnt}
\usepackage{url}
\usepackage{aliascnt}
\usepackage[nameinlink,capitalise]{cleveref}
\theoremstyle{plain}
\newtheorem{thm}{Theorem}[section]

\newaliascnt{lem}{thm}
\newtheorem{lem}[lem]{Lemma}
\aliascntresetthe{lem}

\newaliascnt{cor}{thm}
\newtheorem{cor}[cor]{Corollary}
\aliascntresetthe{cor}

\newaliascnt{prop}{thm}
\newtheorem{prop}[prop]{Proposition}
\aliascntresetthe{prop}

\newaliascnt{conj}{thm}

\aliascntresetthe{conj}

\theoremstyle{definition}
\newaliascnt{ass}{thm}

\aliascntresetthe{ass}

\newaliascnt{rem}{thm}
\newtheorem{rem}[rem]{Remark}
\aliascntresetthe{rem}

\newaliascnt{claim}{thm}

\aliascntresetthe{claim}

\newaliascnt{defn}{thm}

\aliascntresetthe{defn}

\newaliascnt{prob}{thm}
\newtheorem{prob}[prob]{Problem}
\aliascntresetthe{prob}

\newaliascnt{que}{thm}

\aliascntresetthe{que}

\newaliascnt{ex}{thm}

\aliascntresetthe{ex}
\numberwithin{equation}{section}

\def\F{{\mathbb F}}
\def\Q{{\mathbb Q}}
\def\R{{\mathbb R}}
\def\Z{{\mathbb Z}}
\def\C{{\mathbb C}}

\def\GL{\mathop{\mathrm{GL}}\nolimits}

\def\Tr{\mathop{\mathrm{Tr}}\nolimits}

\def\rank{\mathop{\mathrm{rank}}\nolimits}

\def\diag{\mathop{\mathrm{diag}}\nolimits}

\def\poly{\mathop{\mathrm{poly}}}
\newcommand{\SharpP}{\ensuremath{\#\mathrm{P}}}
\newcommand{\SP}{\operatorname{SP}}

\def\ord{\mathrm{ord}}

\def\exp{\mathop{\mathrm{exp}}\nolimits}

\newcommand{\bbE}{\mathbb{E}}

\newcommand{\calA}{\mathcal{A}}

\mathtoolsset{showonlyrefs=true}

\allowdisplaybreaks[4]

\newcommand{\defeq}{\vcentcolon=}

\usepackage{comment}

\newcommand{\eps}{\varepsilon}
\newcommand{\abs}[1]{\left|#1\right|}

\newcommand{\X}{\widehat{\F_q^\times}} % multiplicative character group
\newcommand{\triv}{\mathbf{1}}         % trivial character
\newcommand{\Gauss}{G}                 % Gauss sum symbol

\usepackage[T1]{fontenc}
\usepackage[utf8]{inputenc}
\usepackage{lmodern}
\usepackage{microtype}
\usepackage{amsmath,amssymb,amsfonts,amsthm,mathtools}
\usepackage{hyperref}
\usepackage{aliascnt}
\usepackage[nameinlink,capitalise]{cleveref}

\crefname{thm}{Theorem}{Theorems}
\Crefname{thm}{Theorem}{Theorems}
\crefname{lem}{Lemma}{Lemmas}
\Crefname{lem}{Lemma}{Lemmas}
\crefname{cor}{Corollary}{Corollaries}
\Crefname{cor}{Corollary}{Corollaries}
\crefname{prop}{Proposition}{Propositions}
\Crefname{prop}{Proposition}{Propositions}
\crefname{ass}{Assumption}{Assumptions}
\Crefname{ass}{Assumption}{Assumptions}
\crefname{rem}{Remark}{Remarks}
\Crefname{rem}{Remark}{Remarks}
\crefname{claim}{Claim}{Claims}
\Crefname{claim}{Claim}{Claims}
\crefname{prob}{Problem}{Problems}
\Crefname{prob}{Problem}{Problems}
\crefname{defn}{Definition}{Definitions}
\Crefname{defn}{Definition}{Definitions}

\crefname{appendix}{appendix}{appendices}
\Crefname{appendix}{Appendix}{Appendices}

\usepackage{enumitem}
\usepackage{algorithm}
\usepackage{algorithmic}
\usepackage{booktabs}

\hypersetup{
  colorlinks=true,
  linkcolor=blue,
  citecolor=blue,
  urlcolor=blue
}

\usepackage{braket}
\usepackage{authblk}

\title{Quantum Algorithms and Hardness for Point-Count Approximation over Finite Fields}

\author{Yota Maeda*}
\author{Hiroshi Yano*}

\affil{Toyota Central R\&D Labs., Inc.}

\date{}

\begin{document}

\maketitle
\begingroup
\renewcommand{\thefootnote}{}
\footnotetext{\texttt{\{yota.maeda, hyano\}@mosk.tytlabs.co.jp}, * denotes equal contribution}
\endgroup
\thispagestyle{empty}
\setcounter{page}{0}

\begin{abstract}
We study the approximation of the number of solutions of Laurent polynomials over finite fields.  For a Laurent polynomial
\[
 f(x)=\sum_{j=1}^{s}a_jx^{u_j}\in \F_q[x_1^{\pm1},\ldots,x_n^{\pm1}],
\]
let $U$ be its augmented support matrix whose columns are $(1,u_j)$ with rank $\rho$ and $N(f) := \# \{x\in (\F_q^\times)^n \mid f(x)=0\}$ be its torus point count.
Our first main result is a quantum algorithm that outputs $\widehat{N}(f)$ satisfying 
\[
|\widehat{N}(f) - N(f)| \le \eps q^{n+s/2-\rho}
\]
with success probability $1-\delta$.
Provided that $\rho$ and $\|U\|_\infty$ are bounded, the algorithm runs in both classical bit and quantum gate complexity $\poly(n, s, \log q, 1/\eps, \log(1/\delta))$.
It provides finer resolution than relative-error approximations in general settings.
To the best of our knowledge, in the explicit finite-field input
model considered here, no previous algorithm achieves this
additive accuracy with running time polynomial in $\log q$.
Van Dam (arXiv:quant-ph/0405081) conjectured the existence of such an algorithm under the assumption of an oracle reflecting the algebraic properties of the polynomial. 
In contrast, by exploiting a point-counting formula derived from character sums over finite fields, we develop an alternative approach that efficiently approximates the number of points without assuming the existence of such an oracle.
As a second main result, we prove that the same approximation problem becomes $\#$P-hard under randomized polynomial-time Turing reductions when the support matrix $U$ varies freely as part of the input. 
Thus, taken together, our results clarify how the effectiveness of the quantum approach depends on the tradeoff between the accuracy scale and the support parameters of the input polynomial.
\end{abstract}

\clearpage

\section{Introduction}\label{sec:intro}
Let $q=p^r$ be a prime power and let $\F_q$ be the finite field of $q$
elements. 
A
classical geometric baseline of the number of solutions of a polynomial is the Lang--Weil estimate \cite{LangWeil1954}; if
$X$ is a geometrically integral projective variety of dimension $n-1$, then the number of $\F_q$-points $\#X(\F_q)$ on $X$ satisfies
\begin{align}
    \label{eq:LW}
      \#X(\F_q)=q^{n-1}+O(q^{n-3/2}).
\end{align}
Based on this observation, algorithmic research has advanced to extract more detailed information.  Wan asked for algorithms computing the number of points of smooth projective hypersurfaces with complexity polynomial for 
$\log q$ \cite[Problem 4.2]{wan2008algorithmic}.
Except for the case where the number of variables is small, most notably, the case of curves \cite{schoof1995counting,Kedlaya2001}, no algorithm is currently known that counts the number of solutions of polynomials over $\F_q$ with bit complexity of a polynomial in $\log q$. Complexity-theoretic studies have shown that counting the number of solutions to polynomial equations over finite fields is $\#$P-complete \cite{vonZurGathenKarpinskiShparlinski1997,Milovanov2019}, suggesting that exact point counting is computationally intractable in general structure-free settings.

\noindent
\textbf{Our results.}
In this paper, we investigate both the potential of quantum algorithms and the hardness of problems in approximate point counting over finite fields. From an algorithmic viewpoint, we present a quantum algorithm that approximates the number of torus points with bit and gate complexity $\poly(\log q)$ for a bounded family of polynomials. From the complexity-theoretic viewpoint, we show that even this approximation problem is $\#$P-hard under randomized polynomial-time Turing reductions in general for a uniform input. 

To this end, we introduce the approximation problem in \Cref{prob:torus-approx}, which can be viewed as an intermediate problem between \eqref{eq:LW} and Wan's benchmark problem \cite[Problem~4.2]{wan2008algorithmic}.
Let $n\ge 1$.
A Laurent polynomial on the $n$-dimensional torus is
\[
f(x)\;=\;\sum_{j=1}^s a_j x^{u_j}\in \F_q[x_1^{\pm1},\dots,x_n^{\pm1}],
\qquad a_j\in\F_q^\times,\ u_j\in\Z^n.
\]
Let $U$ be the augmented support matrix whose columns are $(1,u_j)$, and let
$\rho=\rank(U)$.
We write
\[
N(f) := \# \{x\in (\F_q^\times)^n \mid f(x)=0\}.
\]
As the character expansion formula (\Cref{prop:main_character-formula}) shows that the fluctuation term of $N(f)$, that is the second term of \eqref{eq:main-character}, is of order  $q^{n+s/2-\rho}$, the following problem arises naturally.
\begin{prob}
\label{prob:torus-approx}
Given a finite field $\F_q$, integers $n,s\ge 1$, a list $\bigl((u_1,a_1),\ldots,(u_s,a_s)\bigr)$ where
$u_j\in\Z^n$ and $a_j\in\F_q^\times$, which defines a Laurent polynomial $f$, and  parameters $0<\eps,\delta<1$,
 output a number $\widehat N(f)\in\Q$ such that
\[
\Pr\!\left[
  \left|\widehat N(f)-N(f)\right|
  \le \eps q^{n+s/2-\rho}
\right]
\ge 1-\delta.
\]
\end{prob}

Our first main result is the following.
\begin{thm}\label{thm:intro}
Fix constants $r_0,C_0 >0$. On the promise class of inputs satisfying $\rho \leq r_0, \|U\|_\infty\leq C_0$, \Cref{prob:torus-approx} can be solved by a quantum algorithm with a classical bit and quantum gate complexity of $\poly
_{r_0,C_0}(n,s,\log q,1/\eps, \log 1/\delta)$.
\end{thm}

\begin{cor}\label{cor:intro}
For every fixed augmented support matrix $U$, \Cref{prob:torus-approx} can be
solved by a quantum algorithm with a classical bit and quantum gate complexity of 
$
  \poly_U\!\left(\log q, 1/\eps,\log1/\delta\right).
$
\end{cor}
The precise complexity is proved later in \Cref{subsec:complexity-analysis}.
As discussed
below, existing point-counting algorithms do not appear to achieve this error scale with $\poly(\log q)$ bit complexity.

We also prove that this approximation problem becomes $\#$P-hard under randomized polynomial-time Turing reductions when the support matrix $U$ is treated as part of the input. 

\begin{thm}
\label{thm:main-support-uniform-hardness}
Fix a constant $0<\eps<1/4$.  Then \Cref{prob:torus-approx} with $\delta = 1/3$
is $\SharpP$-hard under randomized polynomial-time Turing reductions. 
\end{thm}

The proof of \Cref{thm:main-support-uniform-hardness} is given in
\Cref{sec:computational_complexity}.
This highlights a support-parameter tradeoff in the approximation problem:
bounded support structure enables efficient quantum approximation, while
unrestricted support variation already yields $\#\mathrm P$-hardness at the
same additive accuracy scale.

\noindent
\textbf{Related work.}
We study torus point counting for Laurent polynomials. On the one hand, this problem is considered important for studying mirror symmetry in mathematical physics and algebraic geometry \cite{Kasprzyk2022laurent,KrawitzThesis}.
It has been studied primarily using geometric methods such as cohomology; see the item C below.
On the other hand, algorithmic research has focused on root counting for ordinary polynomials (items A and B below). This distinction is important because the torus restriction allows us to use
multiplicative characters, while ordinary affine root counting has a different
algorithmic structure; we discuss the difference between these two problem settings in \Cref{sec:conclusion}.

\noindent
\textbf{A. Approximation algorithms over finite fields.}
In the typical hypersurface regime, the point count $N(f)$ is of order $q^{n-1}$.
There are several studies on relative approximation \cite{grigoriev1991approximation,karpinski1991lhotzky,karpinski1993approximating} with polynomial-time dependence on $q$.
Over prime fields, some works \cite{huang1996solving,HuangWong1998,Williams2018Counting} provide a relative approximation that runs in  $\poly(\log p, 1/\eps)$.

\noindent
\textbf{B. Algorithms for polynomial systems over finite fields.}
A related line of work studies ordinary polynomial systems over finite fields with the number of variables as the main asymptotic parameter. For bounded-degree systems, faster-than-exhaustive algorithms are known for solving and exact counting, together with reductions and fine-grained lower bounds for low-degree root counting \cite{lokshtanov2017beating,Williams2018Counting,dell2025solving}. Over $\F_2$, further improvements use parity-counting based self-reductions \cite{bjorklund2019solving,dinur2021improved}. These aim to improve over the $q^n$ exhaustive-search baseline and generally retain a power dependence on $q$.

\noindent
\textbf{C. Geometric methods.}
 The Lang--Weil estimate
and its effective refinements imply an estimation as \eqref{eq:LW} \cite{LangWeil1954,GhorpadeLachaud2002,CafureMatera2006,Slavov2023}.
Point counting for Laurent polynomials and toric hypersurfaces has also been studied through $p$-adic cohomology \cite{AdolphsonSperber1989,AdolphsonSperber1990,LauderWan2008,Harvey2015, SperberVoight2013,CostaHarveyKedlaya2019}. In computational directions, these compute the full zeta function, which is stronger than a single point count, and require polynomial complexity in $q$. For curves and abelian varieties, exact point counts can be computed with polynomial dependence on $\log q$ \cite{Schoof1985,schoof1995counting,Pila1990,Kedlaya2001}. 

\noindent
\textbf{D. Quantum algorithms.}
Quantum algorithms for Gauss-sum estimation provide the main primitive
used in our algorithm; for a multiplicative character over $\mathbb F_q$,
the phase of the normalized Gauss sum can be estimated with $\poly(\log q)$ gates \cite{van2003quantum}.  
Under the assumption of the existence of an efficiently implementable unitary access, trace
estimation yields an additive approximation to the point count, and this
program can be carried out explicitly for Fermat-type hypersurfaces \cite{dam2004Quantum}.
However, the construction of such a spectral unitary is not known in
general. 
For curves, there are quantum algorithms that compute the full zeta
function, and hence the exact point count, with $\poly(\log q)$ gates \cite{Kedlaya2006}. 

\noindent
\textbf{E. Computational hardness.}
Exact counting over finite fields remains computationally intractable, at least to the extent of $\#$P-completeness, even when the input is restricted to sparse polynomials \cite{vonZurGathenKarpinskiShparlinski1997,Milovanov2019,ChengHillWan2013}. These results show that sparsity alone is not a tractability assumption once the monomial support is part of the input. However, the $\#$P-completeness of exact counting does not by itself imply hardness of approximation. Indeed, for several standard counting problems, such as counting satisfying assignments of DNF formulas \cite{KarpLubyMadras1989} and computing the permanent of a matrix \cite{Valiant1979Permanent,JerrumSinclairVigoda2004}, the exact version is $\#$P-complete whereas randomized approximation is known to be tractable. This distinction is particularly important over finite fields. Although counting the zeros of polynomial systems can be $\#$P-complete even for low-degree instances, relative-error approximation algorithms are known in several restricted regimes such as fixed or small finite fields \cite{karpinski1993approximating,grigoriev1991approximation,HuangWong1998}. Therefore, establishing approximation hardness requires an argument that explicitly exploits the prescribed approximation guarantee, rather than relying solely on the $\#$P-hardness of a related exact counting problem.

Approximate counting has been extensively studied in computational complexity theory. One standard framework is that of approximation-preserving reductions, which compare approximation problems while preserving the existence of an FPRAS \cite{DyerGoldbergGreenhillJerrum2004,DyerGoldbergJerrum2010}. Another line of work establishes hardness by treating the approximation algorithm as an oracle and proving that such an oracle suffices to recover exact counting or other discrete $\#$P-hard information \cite{GalanisGoldbergHerreraPoyatos2022,GoldbergGuo2017ComplexIsingTutte,Kuperberg2015JonesApprox}. In these results, the assumed approximation oracle is shown to be powerful enough to compute $\#$P-hard information.

\vspace{3mm}
\noindent
\textbf{Techniques.}
The starting point of this paper is to avoid treating point
counting as enumeration over the torus.  A naive approximation strategy
would estimate $N(f)$, whose typical order is $q^{n-1}$, by relative
error.  
While most previous studies achieve approximations of this accuracy, to obtain a more accurate approximation, we derive a character-based formula (\Cref{prop:main_character-formula}) to clarify the contribution of fluctuation terms.  
Using the character orthogonality for finite abelian groups, we express
this fluctuation as a finite Gauss-sum expansion.  The natural size of
this fluctuation is $q^{n+s/2-\rho}$.
In this
expansion, the contributing multiplicative characters must satisfy the relations $\mathcal A_U(q)$ imposed by the
augmented support matrix $U$. 

In general, however, the number of character tuples satisfying these
relations $\#\mathcal A_U(q)$ can still grow exponentially in $\log q$.  Therefore, even if
one could evaluate each summand efficiently, summing all terms would not
give a poly$(\log q)$-complexity algorithm. We overcome this difficulty by giving an explicit parametrization based on the theory of Smith normal forms, which enables us to sample characters uniformly.
This yields the normalized formula $N(f)
  =
  (q-1)^n/q
  +
  q^{s/2-1}(q-1)^{n+1-\rho}B_U(q)\,\Xi_f$,
where $\Xi_f$ is a normalized average over $K_U(q)$ of bounded summands
$T(k)$ satisfying $|T(k)|\le 1$.
Here $K_U(q)$ is the congruence kernel
parametrizing the admissible character tuples, $B_U(q)$ is the Smith-normal-form
factor controlling the size of this kernel, and $T(k)$ is the normalized
Gauss-sum product associated with $k$.  The formal definitions are given in
\Cref{subsec:normalized-character-average}.
Thus the original point-counting problem over a huge torus is converted into the estimation of the mean of a bounded random variable.

To evaluate the normalized average $\Xi_f$ with $\poly(\log q)$ complexity, we leverage quantum algorithms for estimating Gauss sums and multiplicative characters \cite{van2003quantum,dam2004Quantum}.
For each sample $k \in K_U(q)$, the corresponding summand is composed of a product of Gauss sums and the multiplicative characters evaluated at the polynomial coefficients $a_j$.
Following the quantum framework of van Dam \cite{van2003quantum,dam2004Quantum}, these phases can be estimated to polynomial precision using $\poly(\log q)$ quantum gates via the quantum Fourier transform.
Crucially, this polynomial-precision estimation task is believed to be classically intractable, as the discrete logarithm problem reduces to it.
We implement this $k$-dependent quantum primitive within a Hadamard test to manifest each $T(k)$ as a measurable random variable.
Finally, the algorithm classically draws uniform samples of $k$ from $K_U(q)$ using our parametrization via Smith normal forms, applies a Monte Carlo averaging over the outcomes, and the affine conversion above gives the desired additive approximation to $N(f)$.

A further technical contribution is the explicit accounting of all error sources in our algorithm. 
To achieve this, we systematically decouple the algorithmic inaccuracies into two distinct categories: approximation errors and failure probabilities.
The approximation error stems from the approximate implementation of the quantum Fourier transforms over $\Z/N\Z$ for arbitrary $N$ and the statistical estimation error of Monte Carlo sampling.
Conversely, the overall failure probability is distributed across three independent probabilistic subroutines: finding a generator in $\mathbb F_q^\times$, preparing the multiplicative character (chi) states, and the Monte Carlo sampling over $K_U(q)$.
By rigorously allocating the precision parameters and failure probabilities across these components, we ensure the combined estimator meets the desired global guarantee.

Finally, the same transformation explains the complexity boundary proved
in this paper.  
When $U$ is allowed to vary as part of the input, the same additive accuracy is fine enough to recover an embedded
$\#\mathrm P$-hard exact count by rounding (\Cref{thm:toric_hardness_en}).
In the proof, the order $q^{n/2}$ is critical; worse approximation like $q^n$ does not give a similar proof.
Although exact counting and decision problems over finite fields have been extensively studied, to the best of our knowledge there has been no previous result showing that approximating point counting over finite fields is itself $\#$P-hard under some reductions.

\section{Preliminaries}
Fix the additive character
$\Theta:\F_q\to \C^\times$ defined by 
$\Theta(x):=\exp\!\left(2\pi i\Tr_{\F_q/\F_p}(x)/p\right)$
and let $\widehat{\F_q^\times}$ be the multiplicative character group.
For $\chi\in\X$ define the Gauss sum $G(\chi):=\sum_{t\in\F_q^\times}\chi(t)\Theta(t)$.
We recall the basic properties of characters and Gauss sums in
\Cref{app:Gauss_sums}.

Throughout the paper, for functions of several parameters, we write
$\widetilde O(f(x_1,\ldots,x_m))$
to mean $O\!\left(
        f(x_1,\ldots,x_m)
        \cdot
        \log^C(2+x_1+\cdots+x_m)
    \right)$
for some $C>0$.
Equivalently, $\widetilde O$ suppresses factors polylogarithmic in the
relevant input parameters.

We assume that the $n$-bit integer multiplication can be done in $O(n\log n)$ bit complexity by \cite{harvey2021Integer}.
We represent $\mathbb F_q$ as $\mathbb F_p[T]/(\varphi)$ with $\deg\varphi=r$ and use a polynomial basis representation. 
Throughout this subsection, we assume the standard bit complexity model for finite-field arithmetic.
Using fast integer multiplication \cite{harvey2021Integer} and polynomial arithmetic \cite{vonZurGathenGerhard2013}, one field addition, multiplication, and inversion in $\mathbb F_q$ can be performed in $\widetilde O(\log q)$ bit operations. 
For coherent quantum implementations, we assume reversible implementations of these finite-field arithmetic routines with the same cost up to polylogarithmic overhead.
For quantum gate complexity, we consider a gate set consisting of all single-qubit gates and CNOT gates.

\section{Approximate point counting}\label{sec:approximate-point-counting}

\subsection{Normalized character average}\label{subsec:normalized-character-average}
We begin with the character relation formula, which identifies the fluctuation term in the asymptotic expansion.
\begin{prop}
\label{prop:main_character-formula}
The torus point count satisfies
\begin{equation}\label{eq:main-character}
N(f)
=
\frac{(q-1)^n}{q}
\;+\;
\frac{(q-1)^{n+1-s}}{q}
\sum_{(\chi_1,\dots,\chi_s)\in \calA_U(q)}
\left(\prod_{j=1}^s \Gauss(\chi_j^{-1})\,\chi_j(a_j)\right).
\end{equation}
\end{prop}
The proof is given in \Cref{app:prelim}.
Here we introduce 
\[
\calA_U(q):=
\left\{(\chi_1,\dots,\chi_s)\in \X^s:
\prod_{j=1}^s \chi_j=\triv,
\quad
\prod_{j=1}^s \chi_j^{u_{ij}}=\triv\ \text{for }1\le i\le n
\right\}.
\]
 and call it the \emph{admissible character-index set attached to $U$}.
We now rewrite the character expansion in a normalized form suited to sampling and quantum phase estimation.  Fix a generator $g\in\F_q^\times$, and write $\chi$ for the corresponding generator of $\X$. 

Let $\overline{U}:(\Z/(q-1)\Z)^s\to (\Z/(q-1)\Z)^{n+1}$ be the induced map defined by $U$ and we denote by $K_U(q)\defeq \ker(\overline{U})$.
By the theory of the Smith normal form (SNF), detailed in \Cref{app:smith_normal_form}, there exist $P\in\GL_{n+1}(\Z)$, $Q\in\GL_s(\Z)$, and $d_i>0$ for $1\le i\le\rho$ so that
\begin{equation}
    \label{eq:smith normal form of U}
    P U Q=\diag(d_1,\ldots,d_\rho,0,\ldots,0),
\qquad
d_1\mid d_2\mid\cdots\mid d_\rho.
\end{equation}
Define $B_U(q)\defeq\prod_{i=1}^{\rho}g_i$ where $g_i\defeq\gcd(d_i,q-1)$.
As an application of the Smith normal form (\Cref{app:smith_normal_form}), one has the cardinality formula
\begin{equation}\label{eq:K-cardinality}
\# K_U(q)=(q-1)^{s-\rho}B_U(q).
\end{equation}
For $k=(k_1,\ldots,k_s)\in K_U(q)$, put $z(k)\defeq\#\{j:k_j=0\},
\Phi(k)\defeq \prod_{\substack{1\le j\le s\\ k_j\ne 0}}(\Gauss(\chi^{-{k_j}})/\sqrt q)\chi^{k_j}(a_j)$, and 
$T(k)\defeq(-1)^{z(k)}q^{-z(k)/2}\Phi(k)$.
Then $|\Phi(k)|=1$, $|T(k)|=q^{-z(k)/2}\le 1$, and
\begin{equation}\label{eq:product-to-T}
\prod_{j=1}^{s}\Gauss(\chi^{-k_j})\chi^{k_j}(a_j)
=
q^{s/2}T(k).
\end{equation}
Since the terms indexed by $k$ and $-k$ are complex conjugates, the average of $T(k)$ equals the average of $\operatorname{Re}T(k)$; see \Cref{app:Gauss_sums}.
Based on this, define
\begin{equation}\label{eq:Xi-f-def}
\Xi_f\defeq\frac{1}{\# K_U(q)}\sum_{k\in K_U(q)}T(k).
\end{equation} 
Combining \Cref{prop:main_character-formula} and \eqref{eq:K-cardinality}, we obtain
\begin{equation}\label{eq:point-count-conversion}
N(f)
=
\frac{(q-1)^n}{q} + q^{s/2-1}(q-1)^{n+1-\rho}B_U(q)\Xi_f.
\end{equation}

%%%%%%%%%%%%%%%%%%%%%%%%%%%%%%
%%%%%%%%%%%%%%%%%%%%%%%%%%%%%%
\subsection{Algorithm}\label{subsec:algorithm}
It follows from the discussion in \Cref{app:prelim} that $\Xi_f = \bbE_{k\sim\mathrm{Unif}(K_U(q))} \left[\operatorname{Re}T(k)\right]$.
We estimate this expectation by Monte Carlo sampling from $K_U(q)$.
The SNF sampler is given in \Cref{app:smith_normal_form} for this purpose.
Below we describe a classical-quantum hybrid algorithm for constructing an estimator $\widehat N(f)$.
The failure probabilities are allocated as follows: $\delta_{\rm MC}$ is assigned to Monte Carlo sampling, $\delta_{\rm gen}$ to generator finding, and $\delta_\chi$ to chi state preparations. We choose them so that $\delta_{\rm MC}+\delta_{\rm gen}+\delta_\chi\le\delta$.
The parameters $\eps_{\rm ph}$ and $\eps_{\rm seed}$ denote the implementation precision of the phase oracle and the state-preparation accuracy, respectively.
\begin{algorithm}[H]
\caption{\textsc{ApproximatePointCounter}$(\F_q,f,\eps,\delta)$}\label{alg:snf-hadamard-estimator}
\begin{algorithmic}[1]
    \REQUIRE A finite field $\F_q$, a Laurent polynomial $f(x)=\sum_{j=1}^{s}a_jx^{u_j}$, an accuracy parameter $\eps\in(0,1)$, and a failure probability $\delta\in(0,1)$.
    \STATE Set the augmented exponent matrix $U\in M_{n+1,s}(\Z)$.
    \STATE Compute the Smith-normal-form sampling data
    $\mathsf{SNF}_{U,q}$ and compute $B_U(q)$.
    \STATE Set $\delta_{\rm MC}=\delta_{\rm gen}=\delta_\chi=\delta/3$, and compute $\eta_\Xi= \min\left\{1, \eps\left(\frac{q}{q-1}\right)^{n+1-\rho}\frac1{B_U(q)} \right\}$ and $M=\left\lceil\frac{8}{\eta_\Xi^2}\log\frac{2}{\delta_{\rm MC}}\right\rceil$.
    \STATE Find a generator $g\in\F_q^\times$ with a failure probability at most $\delta_{\rm gen}$.
    \STATE Set $\eps_{\rm ph}=\eps_{\rm seed}=\eta_\Xi/6$. (More generally, it suffices to take $\eps_{\rm ph}+2\eps_{\rm seed} \le \frac{\eta_\Xi}{2}.$)
    \FOR{$i=1,\ldots,M$}
    \STATE Run $k_i \leftarrow$ \nameref{alg:snf-sampler}$(\mathsf{SNF}_{U,q})$.
    \STATE Compute $z(k_i)=\#\{j:(k_i)_j=0\}$.
    \STATE Compute the constants $\lambda_j(k_i)\defeq a_j^{-(k_i)_j}\in\F_q^\times \, \text{for all }j\text{ with }(k_i)_j\ne 0.$
    \STATE Run
        $
        Y(k_i)
        \leftarrow$
        \nameref{alg:hadamard-test}$\left(
        \F_q,g,k_i,z(k_i),\{\lambda_j(k_i)\}_{(k_i)_j\ne0},
        \eps_{\rm seed},\eps_{\rm ph},\delta_\chi/M
        \right).
        $
    \ENDFOR
    \STATE Set $\widehat\Xi_f = \frac1M\sum_{i=1}^{M} Y(k_i)$.
    \STATE \textbf{Return} $\widehat N(f) = \frac{(q-1)^n}{q} + q^{s/2-1}(q-1)^{n+1-\rho}B_U(q)\widehat\Xi_f$.
\end{algorithmic}
\end{algorithm}
If generator finding or a chi state preparation inside a call to  \nameref{alg:hadamard-test} fails, the subsequent output is not guaranteed; these events are included in the failure probability analysis.
We note that the algorithm is hybrid.  On the one hand, the SNF computation,
the sampling of $k$, and the coefficient arithmetic are classical.
Each iteration first draws a classical sample $k\in K_U(q)$.
On the other hand, the quantum circuit is then used to generate one bounded random variable $Y(k)$.
The quantum Fourier transforms used in the quantum primitives are recalled in \Cref{app:QFT}.

%%%%%%%%%%%%%%%%%%%%%%%%%%%%%%
\subsection{Complexity analysis}\label{subsec:complexity-analysis}
For the complexity estimates only, we use the notation: $H\defeq\max\{1,\|U\|_\infty\},m:=\max\{n+1,s\}$, and    $L:=\left\lceil \log_2(1+q+H)\right\rceil$. 
All costs are reported in the hybrid form $(T_{\rm cl},T_{\rm q})$, where $T_{\rm cl}$ is measured in the classical bit operations while $T_{\rm q}$ is in quantum gates.

We first show the computational complexity for implemented primitives in \Cref{alg:snf-hadamard-estimator}.
\begin{prop}[Implemented primitives]\label{prop:implemented-primitives}
    The following routines are available.
    \begin{enumerate}
    \item SNF preprocessing computes the data needed to sample from $K_U(q)$, and
    computes $B_U(q)$, in expected classical bit complexity
    $\widetilde O(m^6L+\log(1/\eps)+\log (1/\delta))$.
    
    \item Given the precomputed SNF data, one can sample
    $k\sim{\rm Unif}(K_U(q))$ exactly, and compute the classical constants needed
    for the $k$-dependent circuit, in expected classical bit complexity $\widetilde O(s^2\log q+s(\log q)^2)$.
    
    \item A generator of $\mathbb F_q^\times$ can be found with failure probability
    at most $\delta_{\rm gen}$.
    Its cost consists of  $\widetilde O((\log q)^3\log(1/\delta_{\rm gen}))$ classical bit complexity and 
    $\widetilde O((\log q)^3\log(1/\delta_{\rm gen}))$
    quantum gates.
    
    \item For each $k\in K_U(q)$, there is a Hadamard test circuit producing
    $Y(k)\in[-1,1]$ such that, conditioned on successful state preparation $\cal S_\chi$, $\left|
    \mathbb E[Y(k)\mid k, \mathcal{S}_\chi]-\operatorname{Re}T(k)
    \right|
    \le \eta_\Xi/2$.
    With per-call failure probability $\gamma_\chi$, its quantum gate
    complexity is
    \[
    \widetilde O\!\left(
    (\log q+\log(1/\eta_\Xi))^2
      (\log(1/\eta_\Xi)+\log(1/\gamma_\chi))
    +
    s\log q(\log q+\log s+\log(1/\eta_\Xi))^2
    \right).
    \]
    \end{enumerate}
\end{prop}
Full proofs are given in \Cref{app:proof_details}.
The four items correspond respectively to the SNF preprocessing and scalar parameter computation in Steps 1-3 of \Cref{alg:snf-hadamard-estimator}, the classical sampling and controller work in Steps 7-9, the generator finding in Step 4, and the Hadamard test call in Step 10.
Given the above, we give an end-to-end guarantee for estimating $N(f)$ with an explicit computational complexity.
Recall that we defined \[    \eta_\Xi=
    \min\left\{
    1,\,
    \eps\left(\frac{q}{q-1}\right)^{n+1-\rho}\frac1{B_U(q)}
    \right\},\qquad
        M=
    \left\lceil
    \frac{8}{\eta_\Xi^2}\log\frac{6}{\delta}
    \right\rceil\]
    in \Cref{alg:snf-hadamard-estimator}.

\begin{thm}[End-to-end guarantee]\label{thm:end-to-end}
    \Cref{alg:snf-hadamard-estimator} outputs $\widehat N(f)$ satisfying
    \[
    \Pr\left[
    |\widehat N(f)-N(f)|
    \le
    \eps q^{n+s/2-\rho}
    \right]\ge 1-\delta .
    \]
    In particular, total expected classical bit complexity is
    \[
    T_{\rm cl}
    =
    \widetilde O\!\left(
    m^6L
    +
    (\log q)^3\log\frac1\delta
    +
    \eps_B^{-2}\log\frac1\delta
    \bigl(s^2\log q+s(\log q)^2\bigr)
    \right),
    \]
    and the total quantum gate complexity is
    \[
    T_{\rm q}
    =
    \widetilde O\!\left(
    \eps_B^{-2}\log\frac1\delta
    \left[
    (\log q)^2\log\frac1\delta
    +
    s(\log q)^3
    \right]
    \right).
    \]
       Here we introduced the notation $\eps_B\defeq\min\left\{1,\eps/B_U(q)\right\}$ for readability.
\end{thm}
\begin{proof}
Let $\mathcal E_{\rm gen}$ be the event that generator finding succeeds, and
let $\mathcal E_\chi$ be the event that all $M$ fresh chi-state preparations
succeed.  It follows from the chosen failure allocation (Step 3 in \Cref{alg:snf-hadamard-estimator}) that $\Pr[\mathcal E_{\rm gen}^c]\le \delta/3$ and  $\Pr[\mathcal E_\chi^c]\le \delta/3$.
Condition on $\mathcal E:=\mathcal E_{\rm gen}\cap\mathcal E_\chi$.
For every fixed $k$, \Cref{prop:implemented-primitives} (4) gives
\[
\left|
\mathbb E[Y_i\mid k_i=k,\mathcal E]
-\operatorname{Re}T(k)
\right|
\le \frac{\eta_\Xi}{2},
\]
where we denote $Y(k_i)$ by $Y_i$.
Averaging over the independent draw $k_i\sim{\rm Unif}(K_U(q))$, and using $\Xi_f=\mathbb E_{k\sim{\rm Unif}(K_U(q))}[\operatorname{Re}T(k)]$,
we obtain
\[
\left|
\mathbb E[Y_i\mid\mathcal E]-\Xi_f
\right|
\le \frac{\eta_\Xi}{2}.
\]
Hoeffding's inequality gives
\[
\Pr\left[
\left|
\frac1M\sum_{i=1}^M Y_i-\mathbb E[Y_i\mid\mathcal E]
\right|
\ge
\frac{\eta_\Xi}{2}
\ \middle|\ \mathcal E
\right]
\le
2\exp\left(-\frac{M\eta_\Xi^2}{8}\right)
\le
\frac{\delta}{3},
\]
where the chosen parameters (Step 3 in \Cref{alg:snf-hadamard-estimator}) are used in the final inequality.
Thus, conditioned on $\mathcal E$, the Monte Carlo bad event has probability at most $\delta/3$. 
Since, by union bound, $\Pr[\mathcal E^c]\le \Pr[\mathcal E_{\rm gen}^c]+\Pr[\mathcal E_\chi^c]\le 2\delta/3$, the total failure probability is at most $\delta$.
The affine conversion from $\Xi_f$ to $N(f)$ in \eqref{eq:point-count-conversion} then gives
\[
|\widehat N(f)-N(f)|
\le
\eps q^{n+s/2-\rho}.
\]

It remains to account for the resources.  The SNF preprocessing is performed
once, and by \Cref{prop:implemented-primitives} (1) costs $\widetilde O(m^6L+\log(1/\eps)+\log (1/\delta))$ expected classical bit operations. 
In each of the $M$ trials, after the SNF
data have been computed, \Cref{prop:implemented-primitives} (2)
gives the classical cost $\widetilde O\bigl(s^2\log q+s(\log q)^2\bigr)$ for sampling $k_i$, computing $z(k_i)$, and computing the $k_i$-dependent classical constants used by the circuit. 
Hence, incorporating the classical bit complexity for generator finding $\widetilde O((\log q)^3\log(1/\delta_{\rm gen}))$ by \Cref{prop:implemented-primitives} (3) with $\delta_{\rm gen}=\delta/3$, the total classical cost is
\[
\widetilde O\!\left(
\underbrace{m^6L+\log(1/\eps)+\log (1/\delta)}_{\mathrm{SNF\ preprocessing}}
+
\underbrace{(\log q)^3\log\frac3\delta}_{\mathrm{generator\ finding}}
+
\underbrace{
M\bigl(s^2\log q+s(\log q)^2\bigr)
}_{\mathrm{classical\ sampling\ and\ controller\ work}}
\right).
\]
The simplified forms follow from the convention that $\widetilde O$ suppresses polylogarithmic factors in the displayed parameters.

Assume a generator of $\mathbb F_q^\times$ is found once.
With failure probability $\delta_{\rm gen}=\delta/3$, \Cref{prop:implemented-primitives} (3) gives quantum gate complexity $\widetilde O\!\left((\log q)^3\log\frac{3}{\delta}\right)$.
For each of the $M$ trials, we run the Hadamard test circuit with $\eps_{\rm ph}=\eps_{\rm seed}=\eta_\Xi/6, \, \gamma_\chi=\delta_\chi/M=\frac{\delta}{3M}$.
By \Cref{prop:implemented-primitives} (4), the per-trial quantum gate complexity is
\[
\widetilde O\!\left(
(\log q+\log(2/\eta_\Xi))^2
\left(\log(2/\eta_\Xi)+\log\frac{3M}{\delta}\right)
+
s\log q(\log q+\log s+\log(2/\eta_\Xi))^2
\right).
\]
Multiplying this bound by $M$ and adding the one-time generator-finding
cost gives the total quantum gate complexity
\[
\widetilde O\!\left(
\underbrace{
(\log q)^3\log\frac{3}{\delta}
}_{\mathrm{generator\ finding}}
+
\underbrace{
M\left[
(\log q+\log(2/\eta_\Xi))^2
\left(\log(2/\eta_\Xi)+\log\frac{3M}{\delta}\right)
+
s\log q\,
(\log q+\log s+\log(2/\eta_\Xi))^2
\right]
}_{\mathrm{Hadamard\ test\ circuit}}
\right).
\]
The simplified
forms follow from
\[
M=O\!\left(\eta_\Xi^{-2}\log\frac1\delta\right)
=
O\!\left(\eps_B^{-2}\log\frac1\delta\right)
\]
as $\eps_B\le \eta_\Xi$.
\end{proof}

\noindent
\textbf{Proof of \Cref{thm:intro} and \Cref{cor:intro}.}
As $B_U(q)=\prod_{i=1}^{\rho}\gcd(d_i,q-1)
\le
\prod_{i=1}^{\rho}d_i$ and
the product $\prod_i d_i$ is the $\rho$-th determinantal divisor of $U$, it
is bounded by the largest absolute value of a $\rho\times\rho$ minor up to the
standard determinantal-divisor relation.  Hadamard's inequality implies
$\prod_{i=1}^{\rho}d_i
\le
\rho^{\rho/2}\|U\|_\infty^\rho$.
Thus, if $\rho\le r_0$ and $\|U\|_\infty\le C_0$, then $B_U(q)\le r_0^{r_0/2}C_0^{r_0}=O_{r_0,C_0}(1)$,
so $1/\eps_B=O_{r_0,C_0}(1/\eps)$.
Since the rounding of real-valued quantities to rational numbers incurs no additional computational overhead beyond polynomial time, the above discussion with \Cref{thm:end-to-end} gives
\Cref{thm:intro}.
If $U$ is fixed, then $B_U(q)\le O_U(1)$,
and the same theorem gives \Cref{cor:intro}.
\qed

\vspace{3mm}
Van Dam constructed an approximate quantum algorithm, assuming the existence of an oracle that provides spectral data for a polynomial and can be implemented efficiently \cite[Section 6]{dam2004Quantum}. 
The present algorithm instead samples and evaluates the finite character expansion explicitly using quantum Gauss-sum estimation, and therefore removes this spectral-unitary access assumption. 
Moreover, in the general sparse setting $s=\rho = n+1$, our estimation error becomes $q^{(n-1)/2}$, which coincides with the one conjectured in \cite{dam2004Quantum}.

\section{Computational hardness}
\label{sec:computational_complexity}
To prove the hardness statement announced in
\Cref{thm:main-support-uniform-hardness}, we use the following $\SharpP$-complete counting problem. 
For positive integers $a_1,\ldots,a_N$, define
\[
  \SP(a_1,\ldots,a_N)\defeq 
  \#\left\{\sigma\in\{\pm1\}^N:\sum_{i=1}^N a_i\sigma_i=0\right\}.
\]
We call this \emph{signed partition counting}.
This is a sign-vector variant of the counting version of the classical
\textsc{Partition} problem, whose decision version is one of Karp's 
NP-complete problems~\cite{karp2009reducibility}.
For the counting formulation used
below, we show the following reduction theorem.
\begin{thm}
\label{thm:signed_partition_hard_en}
Signed partition counting is $\SharpP$-complete under polynomial-time Turing reductions.
\end{thm}

\begin{proof}
We reduce from \#SUBSET-SUM, which is a standard $\SharpP$-complete
counting problem; see, for example,
\cite[Section~2.1]{ChengHillWan2013}.  An instance of \#SUBSET-SUM
consists of positive integers $b_1,\ldots,b_t$ with $B:=\sum_{i=1}^t b_i$ and a target integer $T$,
and asks for
\[C_T(b_1,\ldots,b_t)
:=
\#\left\{
S\subseteq \{1,\ldots,t\}:
\sum_{i\in S} b_i=T
\right\}.\]
Note that we may assume that $0< T < B$.
We construct, in polynomial time, an instance of signed partition counting.
Consider the positive integers
$a_i\defeq 
b_i$ for $1\leq i\leq t$, $a_{t+1}\defeq B$, and $a_{t+2} \defeq 2T$.
A sign vector $\sigma\in\{\pm1\}^{t+2}$ with
$\sum_j a_j\sigma_j=0$ is equivalently a choice of a submultiset of
$\{b_1,\ldots,b_t,B,2T\}$
whose sum is $1/2(\sum_{i=1}^t b_i+B+2T)=B+T$.

We first observe that any submultiset of
$\{b_1,\ldots,b_t,B,2T\}$ with sum $B+T$ contains exactly one of the two
distinguished elements $B$ and $2T$.  Indeed, if it contains neither, then
its sum is at most $B$, which is strictly smaller than $B+T$.  If it contains both, then its sum is at least $B+2T$, which is
strictly larger than $B+T$. 

Suppose first that the chosen submultiset contains $B$ and does not contain
$2T$.  Then the contribution from the original elements
$b_1,\ldots,b_t$ must be $T$.  Thus these choices are in bijection with
subsets $S\subseteq\{1,\ldots,t\}$ satisfying
$\sum_{i\in S} b_i=T$.
Suppose next that the chosen submultiset contains $2T$ and does not contain
$B$.  Then the contribution from the original elements must be $B-T$.
Taking complements inside $\{1,\ldots,t\}$ gives a bijection between
subsets of $\{1,\ldots,t\}$ of sum $B-T$ and subsets of sum $T$.

Consequently,
\[
\SP(b_1,\ldots,b_t,B,2T)
=
C_T(b_1,\ldots,b_t)+C_{B-T}(b_1,\ldots,b_t)
=
2C_T(b_1,\ldots,b_t).
\]
Therefore an oracle for signed partition counting computes
$C_T(b_1,\ldots,b_t)$ by one oracle call followed by division by 2.
The construction only appends the two integers $B$ and $2T$, whose bit
lengths are polynomial in the bit length of the original instance.  Hence the
reduction is polynomial time, and signed partition counting is
$\SharpP$-complete.
\end{proof}

We now prove hardness for \Cref{prob:torus-approx}.

\begin{thm}
\label{thm:toric_hardness_en}
For every fixed $0<\eps<1/4$, \Cref{prob:torus-approx} with $\delta=1/3$ is $\SharpP$-hard under randomized polynomial-time Turing reductions.
\end{thm}

\begin{proof}
We reduce from signed partition counting.  Let $a_1,\ldots,a_n$ be positive
integers and $A_0:=\sum_{i=1}^n a_i$.
Consider $f_a=\sum_{i=1}^n \bar a_i x_i^h\in\F_q[x_1^{\pm 1},\cdots, x_n^{\pm 1}]$
where $\bar a_i$ denotes the image of $a_i$ in $\F_q$ and $h\defeq (q-1)/2$.
Choose an odd prime $q$ such that $q>\max\{2A_0,10\}$.
Note that by Bertrand's postulate, there is a prime $q$ with
$2A_0<q<4A_0$, and hence $\log q=O(\log A_0)$, which is polynomial
in the bit length of the signed-partition instance.
Such a prime can be found in randomized polynomial time. By the prime number theorem, a random integer in $(2A_0,4A_0)$ is prime
with probability $\Omega(1/\log A_0)$.  Since primality can be tested in
deterministic polynomial time by \cite{agrawal2004primes}, repeated random sampling
finds such a prime in randomized polynomial time.

For any root $(\phi_i)_i$ of $f_a$, set $\sigma_i:=\phi_i^h\in\{\pm1\}$.  Each sign vector
$\sigma\in\{\pm1\}^n$ has exactly $h^n$ lifts to $(\F_q^\times)^n$.
Moreover, since $q>2A_0$, the equality $\sum_i a_i\sigma_i=0$ in $\Z$
is equivalent to  $\sum_i \bar a_i\sigma_i=0$ in $\F_q$.
Therefore $N(f_a)=h^n\,\SP(a_1,\ldots,a_n)$.
Note that the number of variables is $n$, and the number of monomials is $s=n$.  The support matrix is 
\[
  U=
  \begin{pmatrix}
  1&1&\cdots&1\\
  h&0&\cdots&0\\
  0&h&\cdots&0\\
  \vdots&\vdots&\ddots&\vdots\\
  0&0&\cdots&h
  \end{pmatrix}.
\]
The lower diagonal block has rank $n$, so $\rho=n$.  Hence $n=s=\rho$.
Since $h\ge\sqrt q$ and $0<\eps<1/4$, the approximation error assured in \Cref{prob:torus-approx} satisfies $\eps q^{n+s/2-\rho}=\eps q^{n/2}<1/2\cdot h^n$.
Thus an approximation $\widehat N(f_a)$ determines $\widehat N(f_a)/h^n$ with error less than $1/2$ of the integer $\SP(a_1,\ldots,a_n)$, and rounding
recovers the signed partition count exactly.
 Since in the theorem we set $\delta=1/3$,
this gives a randomized polynomial-time reduction with success
probability at least $2/3$. 
By
\Cref{thm:signed_partition_hard_en}, this concludes the theorem.
\end{proof}

\begin{rem}
\label{rem:hardness-instance-outside-promise}
The reduction above intentionally uses a support matrix whose entries grow
with $q$. Indeed, in the hardness construction we have
$h=(q-1)/2$, hence $\|U\|_\infty=h=\Theta(q)$. Therefore this family does
not satisfy the bounded-support promise in \Cref{thm:intro}.
Moreover, for $U$ in \Cref{thm:toric_hardness_en}, 
the Smith invariant factors are $1,h,\ldots,h$. Since $q-1=2h$, this gives $B_U(q)=h^{n-1}$.
Consequently $\eps_B=\min\{1,\eps/B_U(q)\}$ is of order
$\eps h^{-(n-1)}$, and the sample complexity in \Cref{thm:end-to-end}
is polynomial in $B_U(q)$, not polynomial in $\log q$. This is exactly the
support-parameter tradeoff captured by our results.
\end{rem}

\section{Concluding remarks}
\label{sec:conclusion}
In this paper, we identify Laurent polynomials as a class of structures that are particularly amenable to quantum computation. In the problem setting of exact point count, the problem for Laurent polynomials has the same computational complexity order as the one for ordinary polynomials, known as the $\#$Root$_q^1$ problem. However, when considering the approximation problem studied in this paper, there is a priori no reduction between the two problems. Furthermore, one can show that, similarly to \Cref{thm:toric_hardness_en}, the approximate counting problem for ordinary polynomial systems is also $\#$P-hard under randomized polynomial-time Turing reductions (\Cref{app:Similar results for ordinary polynomials}). 
Recent work \cite{dell2025solving} gives fine-grained lower bounds for exact root counting under \#SETH.
It would be interesting to develop analogous lower bounds for approximate point counting, and to investigate possible quantum fine-grained lower bounds under QSETH \cite{aaronson2019quantum}. 
It would be interesting to develop analogous lower bounds for approximate point counting, and to investigate possible quantum fine-grained lower bounds under QSETH.

\subsection*{Declarations}
The authors used ChatGPT by OpenAI to assist with language editing, exposition, and the organization of background and related-work material. All mathematical statements, proofs, citations, and conclusions were reviewed and verified by the authors, who assume responsibility for all content.

\clearpage
\appendix

\section{Character expansion formula}\label[appendix]{app:prelim}

\subsection{Gauss sums}
\label[appendix]{app:Gauss_sums}
We first record the elementary character identities used throughout this
section.
\begin{lem}\label{lem:orth}
\begin{enumerate}
    \item Additive orthogonality holds; for any $y\in\F_q$,
\[
\frac{1}{q}\sum_{w\in\F_q} \Theta(wy)
=
\begin{cases}
1,& y=0,\\
0,& y\ne 0.
\end{cases}
\]
\item Multiplicative orthogonality holds; for any $t\in\F_q^\times$,
\[
\frac{1}{q-1}\sum_{\chi\in \X}\chi(t)
=
\begin{cases}
1,& t=1,\\
0,& t\ne 1.
\end{cases}
\]
\end{enumerate}
\end{lem}

\begin{proof}
The first identity is the orthogonality relation for the additive character
group of $\F_q$; the character $w\mapsto\Theta(wy)$ is trivial exactly when
$y=0$. The second is the corresponding orthogonality relation for the
character group $\widehat{\F_q^\times}$.
See also \cite[Ch.~5, Sec.~1]{LidlNiederreiter1997}.
\end{proof}

It is classical that $\abs{\Gauss(\chi)}=\sqrt{q}$ for nontrivial $\chi$, and $\Gauss(\triv)=-1$; see \cite[Ch.~1]{BerndtEvansWilliams1998}.
A key identity is the multiplicative Fourier expansion of $\Theta$ restricted to $\F_q^\times$:

\begin{lem}\label{lem:add-to-mult}
For any $t\in\F_q^\times$,
\[
\Theta(t)=\frac{1}{q-1}\sum_{\chi\in\X} \Gauss(\chi^{-1})\,\chi(t).
\]
\end{lem}

\begin{proof}
By the definition of the Gauss sum, we have $\Gauss(\chi^{-1})
=
\sum_{u\in\F_q^\times}\chi^{-1}(u)\Theta(u)$.
Hence
\[
\frac{1}{q-1}\sum_{\chi\in\X}\Gauss(\chi^{-1})\chi(t)
=
\frac{1}{q-1}\sum_{\chi\in\X}
\sum_{u\in\F_q^\times}
\Theta(u)\chi^{-1}(u)\chi(t).
\]
Interchanging the two sums gives
\[
\sum_{u\in\F_q^\times}\Theta(u)
\left(
\frac{1}{q-1}\sum_{\chi\in\X}\chi(tu^{-1})
\right).
\]
By multiplicative character orthogonality, the inner average is $1$ if
$tu^{-1}=1$, equivalently $u=t$, and is $0$ otherwise. Therefore the
last expression equals
\[
\sum_{u\in\F_q^\times}\Theta(u)\mathbf{1}_{u=t}
=
\Theta(t).
\]
\end{proof}
As an application of the properties of multiplicative characters, we show a following property for $\Xi_f$ defined in \eqref{eq:Xi-f-def}.
\begin{lem}\label{lem:Xi-f-real}
    $\Xi_f\in\R$ and
    \begin{equation}\label{eq:Xi-f-Re}
        \Xi_f=\frac{1}{\# K_U(q)}\sum_{k\in K_U(q)}\operatorname{Re}T(k).
    \end{equation}
\end{lem}
\begin{proof}
    First we have
    \begin{equation}
        \overline{\Phi(k)} = \prod_{\substack{1\le j\le s\\ k_j\ne 0}}\dfrac{\overline{\Gauss(\chi^{-{k_j}})}}{\sqrt q}\overline{\chi^{k_j}(a_j)}
        = \prod_{\substack{1\le j\le s\\ k_j\ne 0}}\dfrac{\chi^{k_j}(-1)\Gauss(\chi^{{k_j}})}{\sqrt q}\chi^{-k_j}(a_j)
        = \Phi(-k),
    \end{equation}
    where we use $\overline{G(\chi^{-k_j})}=\sum_t\chi^{k_j}(t)\Theta(-t)=\chi^{k_j}(-1)\sum_t\chi^{k_j}(t)\Theta(t)=\chi^{k_j}(-1)G(\chi^{k_j})$ in the second equality and $\prod_j \chi^{k_j}(-1) = \chi^{\sum_j k_j}(-1) = \chi^0(-1) = 1$ in the last equality.
    This readily yields $\overline{T(k)}=T(-k)$ due to $z(-k)=z(k)$.
    Therefore
    \begin{equation}
        \overline{\Xi_f} = \frac{1}{\# K_U(q)}\sum_{k\in K_U(q)}\overline{T(k)}
        = \frac{1}{\# K_U(q)}\sum_{k\in K_U(q)}T(-k)
        = \Xi_f,
    \end{equation}
    where the last equality follows from the fact that if $k \in K_U(q)$, then $ -k \in K_U(q)$, and a map $\iota(k)=-k$ is bijective.
    Thus $\Xi_f \in \R$.
    This readily yields
    \begin{equation}
        \Xi_f = \operatorname{Re}\Xi_f
        = \frac{1}{\# K_U(q)}\sum_{k\in K_U(q)}\operatorname{Re}T(k).
    \end{equation}
\end{proof}

\subsection{Proof of \Cref{prop:main_character-formula}}
Let $f(x)\;=\;\sum_{j=1}^s a_j x^{u_j}$ be a Laurent polynomial.
By additive orthogonality (\Cref{lem:orth} (1)),
\[
\mathbf{1}_{f(x)=0}
=
\frac{1}{q}\sum_{w\in\F_q}\Theta(w f(x)).
\]
Summing over $x\in(\F_q^\times)^n$ yields
\begin{equation}\label{eq:Nstar-exp}
N(f)
=
\sum_{x\in(\F_q^\times)^n}\mathbf{1}_{f(x)=0}
=
\frac{1}{q}\sum_{w\in\F_q}\ \sum_{x\in(\F_q^\times)^n}\Theta(w f(x)).
\end{equation}
The $w=0$ term contributes $(q-1)^n/q$.

For $w\ne 0$ and each monomial term we use \Cref{lem:add-to-mult}:
\[
\Theta(w a_j x^{u_j})
=
\frac{1}{q-1}\sum_{\chi_j\in\X}\Gauss(\chi_j^{-1})\,\chi_j(w a_j x^{u_j}).
\]
Multiplying over $j=1,\dots,s$ and inserting into \eqref{eq:Nstar-exp} gives
\begin{align}
N(f)
&=
\frac{(q-1)^n}{q}
+\frac{1}{q}\sum_{w\in\F_q^\times}\sum_{x\in(\F_q^\times)^n}
\prod_{j=1}^s
\left(\frac{1}{q-1}\sum_{\chi_j\in\X}\Gauss(\chi_j^{-1})\,\chi_j(w a_j x^{u_j})\right)\notag\\
&=
\frac{(q-1)^n}{q}
+\frac{1}{q}\cdot \frac{1}{(q-1)^s}
\sum_{(\chi_1,\dots,\chi_s)\in\X^s}
\left(\prod_{j=1}^s \Gauss(\chi_j^{-1})\,\chi_j(a_j)\right)\cdot S\cdot T,
\label{eq:expanded}
\end{align}
where
\[
S:=\sum_{w\in\F_q^\times}\prod_{j=1}^s \chi_j(w),
\qquad
T:=\sum_{x\in(\F_q^\times)^n}\prod_{j=1}^s \chi_j(x^{u_j}).
\]
Now
\[
S=\sum_{w\in\F_q^\times}\Big(\prod_{j=1}^s \chi_j\Big)(w)
=
\begin{cases}
q-1,& \prod_{j=1}^s \chi_j=\triv,\\
0,& \text{otherwise},
\end{cases}
\]
by multiplicative orthogonality.

For $T$, write $x=(x_1,\dots,x_n)$ and note
\[
\prod_{j=1}^s \chi_j(x^{u_j})
=
\prod_{j=1}^s \chi_j\!\Big(\prod_{i=1}^n x_i^{u_{ij}}\Big)
=
\prod_{i=1}^n \left(\prod_{j=1}^s \chi_j^{u_{ij}}(x_i)\right),
\]
so the sum factorizes:
\[
T=\prod_{i=1}^n \sum_{x_i\in\F_q^\times}\left(\prod_{j=1}^s \chi_j^{u_{ij}}\right)(x_i).
\]
Each factor is $q-1$ if $\prod_j \chi_j^{u_{ij}}=\triv$, else $0$.
Thus $T=(q-1)^n$ exactly when all $n$ relations hold.
Putting this back into \eqref{eq:expanded} yields the conclusion.

\subsection{Smith normal form}
\label[appendix]{app:smith_normal_form}
We recall a standard fact used to count solutions to linear congruences.
Let the \emph{Smith normal form} (SNF) of $U$ over $\Z$ be
\[P U Q = \mathrm{diag}(d_1,\dots,d_{\rho},0,\dots,0),
\qquad d_1\mid\cdots\mid d_{\rho},\]
where $P\in \mathrm{GL}_{n+1}(\Z)$, $Q\in \mathrm{GL}_s(\Z)$ as \eqref{eq:smith normal form of U} \cite{Newman1971SNF,vonZurGathenGerhard2013}.
We give a Smith normal form formula for the set of admissible character indices.

\begin{prop}\label{prop:snf-kernel}
\begin{enumerate}
    \item After choosing a generator of $\X$, we have an isomorphism $\calA_U(q) \cong K_U(q)$.
    \item The order of $K_U(q)$ is 
    \[
\begin{aligned}
(q-1)^{s-\rho}\prod_{i=1}^{\rho} \gcd(d_i,q-1).
\end{aligned}
\]
\end{enumerate}
\end{prop}
\begin{proof}
(1) Fix a generator $g$ of $\F_q^\times$.
Then each $\chi\in\X$ is determined by $\chi(g)=\zeta_{q-1}^k$ for a unique $k\in\Z/(q-1)\Z$.
Hence an $s$-tuple $(\chi_1,\dots,\chi_s)$ corresponds to a vector $k=(k_1,\dots,k_s)\in(\Z/(q-1)\Z)^s$.
The constraints
\[
\prod_{j=1}^s \chi_j=\triv,\qquad
\prod_{j=1}^s \chi_j^{u_{ij}}=\triv \ (1\le i\le n)
\]
become
\[
\sum_{j=1}^s k_j\equiv 0\pmod{q-1},\qquad
\sum_{j=1}^s u_{ij}k_j\equiv 0\pmod{q-1},
\]
i.e.\ $Uk\equiv 0\pmod{q-1}$ for the augmented exponent matrix $U$.

(2) Let $PUQ=D$ be an SNF with $D=\mathrm{diag}(d_1,\dots,d_{\rho},0,\dots,0)$.
Over $\Z/(q-1)\Z$, the matrices $P$ and $Q$ are invertible, so $\ker(\overline{U})\cong \ker(\overline{D})$.
The system $\overline{D} y=0$ means $d_i y_i\equiv 0\pmod{q-1}$ for $i\le \rho$, while the remaining $s-\rho$ coordinates are free.
The number of solutions to $d_i y_i\equiv 0\pmod{q-1}$ is $\gcd(d_i,q-1)$, yielding the formula.
\end{proof}
Thus \Cref{prop:snf-kernel} gives
\[
\# K_U(q)=(q-1)^{s-\rho}B_U(q).
\]
This factor is used in the precise sampling and running-time analysis in \Cref{sec:approximate-point-counting}.
We now present an algorithm to extract elements from $K_U(q)$, which is needed for Monte Carlo sampling.
By \eqref{eq:smith normal form of U}, any element in $K_U(q)$ can be written as $y = (y_i)_{i=1}^s$ with $\diag (d_1,\cdots,d_\rho,0,\cdots,0) y \equiv 0 \pmod{q-1}$.
Combined with the proof of \Cref{prop:snf-kernel}, we obtain  \Cref{alg:snf-sampler}.

\begin{algorithm}[H]
\caption{\textsc{SNFSampler}$(\mathsf{SNF}_{U,q})$}
\makeatletter\protected@edef\@currentlabelname{\textsc{SNFSampler}}\makeatother
\label{alg:snf-sampler}
\begin{algorithmic}[1]
    \REQUIRE Precomputed data
    $
    \mathsf{SNF}_{U,q}
    =
    \left(
    q,\rho,\bar Q,(g_i)_{i=1}^s,(c_i)_{i=1}^s
    \right),
    $
    where $\bar Q=Q\pmod{q-1}$,
    $
    g_i=\gcd(d_i,q-1)\,(1\le i\le \rho),
    \,
    g_i=q-1\,(\rho<i\le s),
    $
    and
    $
    c_i=(q-1)/g_i\,(1\le i\le s).
    $

    \FOR{$i=1,\ldots,s$}
        \STATE Sample independent random integers
        $
        r_i\sim\operatorname{Unif}\{0,1,\ldots,g_i-1\}.
        $
        \STATE Set
        $
        y_i= c_i r_i \pmod{q-1}.
        $
    \ENDFOR

    \STATE Set
    $
    y= (y_1,\ldots,y_s)^\top\in(\Z/(q-1)\Z)^s.
    $

    \STATE Set
    $
    k= \bar Q y \pmod{q-1}.
    $
    
    \STATE \textbf{Return} $k$.
\end{algorithmic}
\end{algorithm}
The map
$
(r_i)_{i=1}^s
\longmapsto
y=(c_1r_1,\ldots,c_sr_s)\pmod{q-1}
$
is a bijection from
$
\prod_{i=1}^s \{0,\ldots,g_i-1\}
$
onto $\ker(\overline D)$, where
$\overline D=\operatorname{diag}(d_1,\ldots,d_\rho,0,\ldots,0)\pmod{q-1}$.
Since $\bar Q$ is invertible over $\Z/(q-1)\Z$, the map $y\mapsto \bar Qy$
is a bijection from $\ker(\overline D)$ to $K_U(q)$.  Therefore
\Cref{alg:snf-sampler} outputs the uniform distribution on $K_U(q)$.
Next, we estimate its computational cost.

\begin{prop}\label{Prop:enum}
 There is a deterministic classical algorithm that computes $\overline P\in GL_{n+1}(\Z/(q-1)\Z)$ and $\overline Q\in GL_s(\Z/(q-1)\Z)$,
together with a diagonal matrix $\overline{D}=\operatorname{diag}(\bar d_1,\dots,\bar d_\rho,0,\dots,0)$
such that $\overline P\,\overline{U}\,\overline Q=\overline{D}$ over $\Z/(q-1)\Z$, with 
$T_{\rm SNF}(U,q)\defeq\widetilde O\!\left(m^6 L\right)$ bit complexity.
After this preprocessing, \Cref{alg:snf-sampler} samples uniformly from $K_U(q)$ with $\widetilde O(s^2\log q)$
bit operations.
\end{prop}

\begin{proof}
The modular Smith reduction over the principal ideal ring $\Z/(q-1)\Z$ is standard.  The existence of Smith forms over principal ideal rings follows from the elementary divisor property, and modular Smith operations are discussed in \cite[Theorem~2.1(4)]{Stanley2016SNF} and \cite[Sec.~2]{Storjohann1996}.
The Smith form of an $(n+1)\times s$ matrix can be computed in polynomial
matrix-arithmetic time.  We use a conservative dense bound $O(m^6)$ ring
operations, including the accumulation of the row and column elementary
transformations giving $\overline{P}$ and $\overline{Q}$.
Using fast integer multiplication and half-gcd algorithms, gcd and extended-gcd
operations on $O(L)$-bit integers cost $\widetilde O(L)$ bit operations \cite[Chs.~3 and 11]{vonZurGathenGerhard2013}.  Applying the usual pivot reduction to the $(n+1)\times s$ matrix $\overline U$ and accumulating the row and column multipliers gives the displayed conservative bound for $T_{\rm SNF}(U,q)$.

It remains to describe the sampler.  Since $\overline P$ is invertible over $\Z/(q-1)\Z$, the condition $\overline U k=0$ is equivalent to $\overline D y=0$ with $k=\overline Qy$.  
The congruence $\bar d_i y_i\equiv0\pmod{q-1}$ has exactly $g_i$ solutions,
\[
y_i=\frac{q-1}{g_i}r_i\pmod{q-1},
\qquad r_i\in\{0,\dots,g_i-1\}.
\]
The remaining coordinates $y_{\rho+1},\dots,y_s$ are free in $\Z/(q-1)\Z$.  Choosing all these parameters independently and uniformly therefore gives a uniform $y$ in $\ker(\overline D)$, and $k=\overline Qy$ is uniform in $K_U(q)$.  The dense multiplication by the $s\times s$ matrix $\overline Q$ costs $\widetilde O(s^2\log q)$ bit operations.
\end{proof}

\section{Quantum Fourier transform}
\label[appendix]{app:QFT}
The classical discrete Fourier transform (DFT) and the quantum Fourier transform (QFT) are represented by essentially the same Fourier matrix, but they are used in different computational models.  
The DFT or QFT transforms a state $f \in \C^N$ into a state $\hat f$, where $\hat f \in \C^N$ is the Fourier transform $\hat f$ of $f$.
While the classical fast Fourier transform (FFT) needs $O(N \log N)$ operations, the QFT can be implemented by $O((\log N)^2)$ quantum gates and is thus very efficient.
However, the QFT is not a direct quantum replacement for the classical FFT: due to the readout limitation, one cannot generally extract the entire list of Fourier coefficients from a single quantum state.  Its algorithmic value instead lies in using interference to extract global information, such as periods, hidden subgroups, or phases encoded in character sums.

The first QFT needed in this paper is the cyclic QFT $F_N$ over $\Z/N\Z$.  When $N$ is a power of two, this is the standard QFT circuit and can be implemented exactly in our gate model using $O((\log N)^2)$ gates \cite{cleve1998Quantum}. 
When $N$ is not a power of two, the situation is more delicate.
Efficient approximate QFTs over arbitrary cyclic groups were studied by
Hales and Hallgren~\cite{hales2000Improved}. 
For the present error analysis, we give below a self-contained operator-norm implementation bound, obtained by combining coherent phase estimation with standard reversible arithmetic. 
The resulting quadratic polylogarithmic bound is sufficient for our application.

The second QFT needed in this paper is the additive QFT over the finite field $\F_q$, where $q=p^r$.
Fixing the nontrivial additive character $\Theta(x)=\omega_p^{\operatorname{Tr}_{\F_q/\F_p}(x)}$, this transform is
\[
    F_{\F_q}\ket{x}
    =
    \frac1{\sqrt q}
    \sum_{y\in\F_q}
    \Theta(xy)\ket{y}.
\]
If $\F_q$ is represented using an $\F_p$-basis $\mathcal B=(b_1,\ldots,b_r)$, then $F_p^{\otimes r}$ implements this transform directly only when $\mathcal B$ is self-dual for the trace pairing.  For a general coordinate representation, one must account for the Gram matrix
\[
    M_{\mathcal B}
    =
    \bigl(\operatorname{Tr}_{\F_q/\F_p}(b_i b_j)\bigr)_{i,j}.
\]
If the QFT output were immediately measured, the change of coordinates induced by $M_{\mathcal B}^{-1}$ could be absorbed into a classical relabeling of the measurement outcome.
In our algorithm, however, the QFT output is subsequently fed into coherent finite-field arithmetic.
Therefore this relabeling must be implemented coherently by a reversible linear map $\ket{z}\mapsto \ket{M_{\mathcal B}^{-1}z}$ so that the output returns to the same finite-field encoding used by the arithmetic circuits.

The role of this section is therefore implementational.
In \Cref{alg:hadamard-test} shown in the following appendix, we follow the Gauss-sum framework of van Dam \cite{van2003quantum,dam2004Quantum}.
The multiplicative-character part of the Gauss sum is supplied by chi states, whose preparation uses the cyclic QFT over $\Z/(q-1)\Z$. 
The additive-character part is supplied by the Gauss-sum phase oracle, whose construction uses the additive QFT over $\F_q$. 
Since both transforms are used coherently inside the quantum subroutine, the bounds below are stated in operator norm: first for the cyclic QFT $F_N$, using a coherent phase estimation implementation, and then for the finite-field QFT $F_{\F_q}$, including the coherent linear change of coordinates associated with $M_{\mathcal B}$.

\begin{prop}[QFT $F_N$ over $\Z/N\Z$]\label{prop:QFT_over_ZN}
    Let $N\ge 2$, and define
    \[
        F_N
        :=
        \frac1{\sqrt N}
        \sum_{x,y=0}^{N-1}
        \omega_N^{xy}\ket{y}\bra{x},
        \qquad
        \omega_N:=e^{2\pi i/N}.
    \]
    For every $\varepsilon\in(0,1)$, there is a quantum circuit that implements $F_N$ to operator-norm error at most
    $\eps$ using
    $
        O\!\left((\log N+\log(1/\eps))^2\right)
    $
    gates.
\end{prop}
\begin{proof}
    Let $b_N = \lceil \log_2 N \rceil$.
    When $N$ is a power of two, it is typically known that the QFT circuit can be implemented exactly using $O(b_N^2)$ gates \cite{cleve1998Quantum}.
    Below, we deal with the case when $N$ is not a power of two.
    
    The construction here follows from \cite{childs2010Quantum} using quantum phase estimation.
    The goal is to implement $\ket{x} \mapsto \ket{\hat x} = F_N \ket{x}$.
    It is split into two steps.
    The first step is $\ket{x}_A\ket{0}_B \mapsto \ket{x}_A\ket{\hat x}_B$ for $b_N$-qubit registers $A$ and $B$.
    The second step is $\ket{x}_A\ket{\hat x}_B \mapsto \ket{0}_A\ket{\hat x}_B$. 
    Noting that a modular adder $U_{\mathrm{add}}:\ket{x}\mapsto\ket{x+1}$ satisfies that $U_{\mathrm{add}}\ket{\hat x}=\omega_N^{-x}\ket{\hat x}$, we can use it for phase estimation to erase the first register $A$.

    We first implement the map $\ket{x}_A\ket{0}_B \mapsto \ket{x}_A\ket{\hat x}_B$.
    Prepare the interval state $\ket{u_N} \defeq (1/\sqrt{N})\sum_{y=0}^{N-1}\ket{y}$ in register $B$.
    This can be done by amplifying the uniform superposition over all $2^{b_N}$ computational basis states into the target subspace $G_y = \{0,...,N-1\}$ via exact amplitude amplification.
    Considering the cost of marking the target subspace, $O(b_N)$ Toffoli gates, and the initial overlap $O(\sqrt{N/2^{b_N}})$, it costs $O(b_N)$ gates.
    Next, apply diagonal controlled phase $V: \ket{x}_A\ket{y}_B \mapsto \omega_N^{xy}\ket{x}_A\ket{y}_B$.
    Writing $x=\sum_{\alpha=0}^{b_N-1}2^\alpha x_\alpha$ and
    $y=\sum_{\beta=0}^{b_N-1}2^\beta y_\beta$, this phase decomposes as
    \[
        \omega_N^{xy}
        =
        \prod_{\alpha,\beta=0}^{b_N-1}
        \exp\!\left(\frac{2\pi i\,2^{\alpha+\beta}}{N}x_\alpha y_\beta\right),
    \]
    so $V$ is implemented by $O(b_N^2)$ controlled phase gates.  Hence
    $
        V\bigl(\ket{x}_A\ket{u_N}_B\bigr)
        =
        \ket{x}_A\ket{\hat x}_B.
    $

    The second step is to erase the first register $A$ coherently.
    This requires two more working registers, $C$ and $D$.
    The register $C$ corresponds to $t$ ancilla qubits for phase estimation, and the other is an $n$-qubit register for rounding the $t$-qubit phase estimation result.
    Then, the register $D$ is used to subtract the decoded value from the first register $A$.
    First, applying phase estimation $U_{PE}$ for a modular adder $U_{\mathrm{add}}$ yields $U_{PE}\ket{\hat x}_B\ket{0^t}_C = \ket{\hat x}_B\otimes(\sum_{j=0}^{2^m-1}\gamma_{x,j}\ket{j}_C)$ with $\gamma_{x,j=x}$ (such that $j$ is the best $m$-bit approximation) being the highest amplitude.
    Note that $(b_N+1)$-bit approximation uniquely determines $x$.
    By the standard result of phase estimation, one can choose $t=b_N+1+\ell_{\rm pe}$ with $\ell_{\rm pe}=\lceil \log(2+(1/(2\delta))) \rceil$ to obtain the best $(b_N+1)$-bit approximation with probability at least $1-\delta$.
    This phase estimation circuit requires $O(tb_N)$ gates for multi-controlled modular adders $U_{\mathrm{add}}^{2^j}$ for $j=0,...,t-1$ and $O(t^2)$ gates for an inverse QFT on the register $C$.
    Let $G_x \subseteq \{0,...,2^t-1\}$ be the set of the nearest $(b_N+1)$-bit approximation $j$, and let
    $
        \sum_{j\in G_x}|\gamma_{x,j}|^2\ge 1-\delta
    $
    be the corresponding success probability.
    Next, we apply the $b_N$-bit rounding operation $U_{\mathrm{round}}: \ket{j}_C\ket{0}_D\mapsto\ket{j}_C\ket{x(j)}_D$ yielding $x(j)=x, \forall j \in G_x$.
    After applying a modular subtraction $U_{\mathrm{sub}}: \ket{x}_A\ket{y}_D\mapsto\ket{x-y}_A\ket{y}_D$ to the $AD$ register state, the resulting state has the form
    \[
        \ket{\Psi_x}
        =
        \sqrt{p_x}\,
        \ket{0}_A\ket{\hat x}_B\ket{0}_{CD}
        +
        \sqrt{1-p_x}\,\ket{E_x}_{ABCD},
        \qquad
        p_x\ge 1-\delta,
    \]
    where 
    $
        \ket{E_x}_{ABCD}
        \perp
        \ket{0}_A\ket{\widehat x}_B\ket{0}_{\rm work}.
    $
    The unitary operators $U_{\mathrm{round}}$ and $U_{\mathrm{sub}}$ require $O(t^2)$ and $O(b_N)$ gates, respectively.
    Applying $U_{\mathrm{round}}^\dagger U_{PE}^\dagger$ reverts the $CD$ registers to $\ket{0}_C\ket{0}_D$.
    Therefore, denoting the above implementation by $\widetilde F_N$, 
    \[
    \begin{aligned}
    &\left\|
        \widetilde F_N\bigl(\ket{\psi}_A\ket{0}_B\ket{0}_{CD}\bigr)
        -
        \ket{0}_A(F_N\ket{\psi})_B\ket{0}_{CD}
     \right\|^2  \\
    &\qquad =
    \sum_x|\alpha_x|^2
    \left\|
        \ket{\Psi_x}
        -
        \ket{0}_A\ket{\hat x}_B\ket{0}_{CD}
    \right\|^2                                      \\
    &\qquad \le
    \sum_x|\alpha_x|^2\,2(1-p_x)
    \le 2\delta,
    \end{aligned}
    \]
    for every normalized $\ket{\psi}=\sum_x\alpha_x\ket{x}$.
    Setting $\delta=\eps^2/2$ guarantees the operator-norm error bounded by $\eps$.
    Consequently, the total gate complexity is $O(b_N+b_N^2+tb_N+t^2)=O(t^2)=O((\log N + \log(1/\eps))^2)$.
\end{proof}

\begin{prop}[Additive QFT over finite fields]\label{prop:finite_field_QFT}
For $\eps\in(0,1)$, there exist a unitary circuit $\tilde F_{\F_q}$ such that
    \begin{equation}
        \|\tilde F_{\F_q} - F_{\F_q}\| \leq \eps,
    \end{equation}
    with gate complexity $O(\log q(\log q + \log(1/\eps))^2)$.
\end{prop}
\begin{proof}
    QFT over $\F_q$ can be implemented by employing QFT over $\F_p$ in parallel as follows.
    Let $\mathcal B=(b_1,\ldots,b_r)$ be a polynomial basis used to represent $\mathbb F_q/\mathbb F_p$, and write
    $
        x=\sum_{i=1}^r x_i b_i,\,
        y=\sum_{i=1}^r y_i b_i.
    $
    Let $M_{\mathcal B}\in M_r(\mathbb F_p)$ be the Gram matrix of the trace pairing,
    $
        (M_{\mathcal B})_{ij}:=\operatorname{Tr}_{\mathbb F_q/\mathbb F_p}(b_i b_j).
    $
    Since the trace pairing is nondegenerate, $M_{\mathcal B}\in GL_r(\mathbb F_p)$.
    Then
    $
        \operatorname{Tr}_{\mathbb F_q/\mathbb F_p}(xy)
        =
        x^\top M_{\mathcal B} y.
    $
    Applying $F_p^{\otimes r}$ maps
    $
        \ket{x}
        \mapsto
        \frac1{\sqrt q}\sum_{z\in\mathbb F_p^r}
        \omega_p^{x^\top z}\ket{z}.
    $
    The reversible linear map 
    $
        L_{\mathcal B}:\ket{z}\mapsto\ket{M_{\mathcal B}^{-1}z}
    $
    then gives
    $
        \frac1{\sqrt q}\sum_{y\in\mathbb F_p^r}
        \omega_p^{x^\top M_{\mathcal B}y}\ket{y}
        =
        F_{\mathbb F_q}\ket{x}.
    $
    Now invoke the QFT over $\Z/N\Z$ with $N=p$.
    From \Cref{prop:QFT_over_ZN}, for $\eps_0 \in (0,1)$, one can implement a unitary $\widetilde F_p$ such that $\|\widetilde F_p - F_p\| \leq \eps_0$ with a gate complexity of $G_p(\eps_0)=O((\log p + \log(1/\eps_0))^2)$.
    Define $\widetilde F_{\F_q} = L_{\mathcal B} \widetilde F_p^{\otimes r}$.
    Use a telescoping expansion $\widetilde F_p^{\otimes r} - F_p^{\otimes r} = \sum_{j=1}^{r} \widetilde F_p^{\otimes(j-1)} \otimes (\widetilde F_p - F_p) \otimes F_p^{\otimes(r-j)}$ to obtain $\|\widetilde F_{\F_q} - F_{\F_q}\| \leq \sum_{j=1}^{r} \|\widetilde F_p\|^{j-1} \cdot \|\widetilde F_p - F_p\| \cdot \|F_p\|^{r-j} \leq r \eps_0$.
    Choosing $\eps_0=\eps/r$ implies that $\|\widetilde F_{\F_q} - F_{\F_q}\| \leq \eps$.
    Thus, the gate complexity to implement $\widetilde F_{\F_q}$ is $r \cdot G(\eps/r) = O(r(\log p + \log(r/\eps))^2) \subseteq O\!\left(\log q\,(\log q+\log(1/\varepsilon))^2\right)$.
    The cost of $L_{\mathcal B}$ is absorbed into the coarse bound.
\end{proof}

\section{Implementation details}
\label[appendix]{app:proof_details}
This appendix supplies the implementation details used in
\Cref{prop:implemented-primitives}.
We separate the analysis into the classical part and the quantum part.
The classical part consists of the one-time Smith-normal-form preprocessing and the per-sample controller work in the Monte Carlo loop.  
The quantum part consists of one-time generator finding and the Hadamard test circuit used to estimate the real part of the normalized Gauss-sum product in the Monte Carlo loop. 
At the end of the appendix, we collect these estimates and derive \Cref{prop:implemented-primitives}.

\subsection{Classical computation costs}
We first analyze the classical costs.  The preprocessing step computes the
Smith-normal-form data needed to sample from the congruence kernel
$K_U(q)$, together with the normalization factor $B_U(q)$ and the scalar
parameters used by the estimator.  This part is performed only once, before
the Monte Carlo loop starts.
Let $\bar d_i:=d_i\pmod{q-1}$ for $1\le i\le \rho$ and $\bar{Q}:=Q\pmod{q-1}$.
\begin{lem}[SNF preprocessing; Steps 1-3 in \Cref{alg:snf-hadamard-estimator}]\label{lem:classical-preprocessing-cost}
    Steps 1, 2, and 3 in \Cref{alg:snf-hadamard-estimator} compute the sampler data
$\mathsf{SNF}_{U,q}
=
\bigl(q,\rho,\bar Q,(g_i)_{i=1}^s,(c_i)_{i=1}^s\bigr)$ with a bit complexity bounded by
        \begin{equation}\label{eq:preprocessing-poly-general}
     T_{\rm pre}^{\rm cl} = \widetilde{O}(m^6L + \log (1/\eps) + \log (1/\delta_{\rm MC})).
        \end{equation}
\end{lem}
\begin{proof}
    \Cref{alg:snf-hadamard-estimator} needs the following  data to instantiate \Cref{alg:snf-sampler}: the invariant factors
    $d_1,\ldots,d_\rho$ only through the quantities $g_i$
    and it requires the right unimodular transformation $Q$ only through its
    reduction modulo $q-1$, because the sampler outputs $k=Qy \pmod{q-1}$.
     We therefore compute and store $\bar d_i,\bar{Q}$.
    This is done in $T_{\rm SNF}^{\rm ker}(U,q)$ bit complexity, the cost of computing the Smith
    data sufficient for sampling, namely the invariant factors modulo $q-1$ and the
    right transformation matrix reduced modulo $q-1$.
    It is bounded by
    $\widetilde{O}(m^6L)$ by \Cref{Prop:enum}.  
    
    Next, for each $1\le i\le \rho$, compute $g_i=\gcd(\bar d_i,q-1)$.
    Both $\bar d_i$ and $q-1$ have bit length at most $O(\log q)$, and hence at
    most $O(L)$.  Using the fast Euclidean algorithm~\cite[Chs.~3 and~11]{vonZurGathenGerhard2013}, combined with \cite{harvey2021Integer}, each gcd computation costs $\widetilde{O}(L)$ bit operations.  Thus the
    total cost of the $\rho$ gcd computations is $O(\rho L)$.
    
    The product $B_U(q)=\prod_{i=1}^{\rho}g_i$
    has bit length at most $\rho\log (q-1)$ since $g_i\le q-1$ for every $i$.
    It can be computed by a balanced product tree.  Its cost
    is bounded by
    \begin{equation}
    \sum_{h=0}^{\lceil\log_2\rho\rceil-1}
    \left\lceil\frac{\rho}{2^{h+1}}\right\rceil
    \widetilde{O}(2^h\log q) = \widetilde O(\rho\log q).
    \end{equation}
Note that the cost of the computation of  $c_i$'s is most $\widetilde O(s\log q)$ bit operations.
    
    The scalar quantities $\eta_\Xi$
    and $M$ have bit lengths $O(n\log q + \rho \log q + \log (1/\eps) + \log (1/\delta_{\rm MC}))$, so 
    one can compute them in $T_{\rm scal} = \widetilde{O}(n\log q + \rho \log q + \log (1/\eps) + \log (1/\delta_{\rm MC}))$ bit operations.
    
    Finally, the remaining classical preprocessing consists of reading and storing
    the coefficients $a_j$ and the finite-field representation data needed by the
    later arithmetic routines.  Under the assumed finite-field arithmetic model,
    this contributes $\widetilde{O}(s\log q)$
    bit operations.  Summing all contributions gives
    \begin{equation}\label{eq:classical-preprocessing-cost}
    T_{\rm pre}^{\rm cl}
    =
    T_{\rm SNF}^{\rm ker}(U,q)
    +
   \widetilde{O}(\rho L)
    +
    \widetilde O(\rho\log q)
    +
    \widetilde{O}(s\log q)
    +
    T_{\rm scal}.
    \end{equation} 
    The bound \eqref{eq:preprocessing-poly-general} follows immediately.
\end{proof}

After the preprocessing data have been computed, each Monte Carlo iteration
requires only classical sampling from the diagonalized kernel, reconstruction
of the vector $k\in K_U(q)$, and the computation of the $k$-dependent
field constants used in the phase oracle.  The following lemma isolates this
per-sample classical controller cost. 
    \begin{lem}[Classical sampling and controller cost; Steps 7-9 in \Cref{alg:snf-hadamard-estimator}]\label{lem:classical-per-sample}
    Given $\mathsf{SNF}_{U,q}$, the expected classical bit complexity of
    one Monte Carlo sample in \Cref{alg:snf-hadamard-estimator}, including the classical computation of
    the fixed-$k$ coefficient constants $a_j^{-k_j}$ used in the phase oracle,
    is
    \begin{equation}\label{eq:classical-per-sample-cost}
    C_{\rm samp}^{\rm cl}
    =
    \widetilde O\!\left(
    s^2\log q
    +
    s(\log q)^2
    +
    \log M
    \right).
    \end{equation}
\end{lem}
\begin{proof}
    We decompose the classical work in one Monte Carlo sample.
    First, we sample $k\sim \operatorname{Unif}(K_U(q))$ by running
    \Cref{alg:snf-sampler}.  This includes drawing the parameters $r_i$,
    forming $y_i=c_ir_i\pmod{q-1}$, and computing $k=\bar Qy\pmod{q-1}$.
    By \Cref{Prop:enum}, after the SNF preprocessing this costs $C_{\rm ker}
        =
        \widetilde O(s^2\log q)$
    expected bit operations.
    
    After $k$ is obtained, the algorithm computes $z(k)$.
    This requires $s$ comparisons modulo $q-1$, each on $O(\log q)$-bit
    integers.  Thus the cost is $C_z
    =
    O(s\log q)$.
    
    Next, for the fixed-$k$ phase oracle in \Cref{lem:quantum-per-sample},
    the classical controller computes the field constants $a_j^{-k_j}\in \mathbb F_q^\times$
    for those $j$ with $k_j\ne 0$, or for all $1\le j\le s$ in the worst
    case.  Each exponent $k_j$ has $O(\log q)$ bits.  By repeated
    squaring in the assumed finite-field representation, each such exponentiation
    costs $\widetilde{O}((\log q)^2)$ bit operations.  Hence the total classical
    cost of computing these coefficient constants is $C_{\rm coeff}
    =
    \widetilde{O}(s(\log q)^2)$.
    They are used in the circuit description of $U_{k_i}$.

    Finally, after the Hadamard test subroutine returns $Y(k)$, the classical controller updates the running sum used to compute $\widehat\Xi_f=M^{-1}\sum_i Y(k_i)$.  This bookkeeping costs $O(\log M)$ bit operations per sample.
    
    Summing the contributions gives
    \begin{align}
    C_{\rm samp}^{\rm cl}
    &=
    C_{\rm ker}
    +
    C_z
    +
    C_{\rm coeff}
    +
    C_{\rm acc} \\
    &=
    \widetilde O\!\left(
    s^2\log q
    +
   s(\log q)^2
    +
    \log M
    \right).
    \end{align}
\end{proof}

\subsection{Quantum primitives}
We now turn to the quantum primitives.  The first primitive is generator
finding in $\F_q^\times$, which is needed to identify multiplicative
characters with exponents modulo $q-1$.  Once a generator is fixed, each
Monte Carlo iteration uses a Hadamard test whose eigenvalue is the normalized
Gauss-sum product $\Phi(k)$.
\begin{lem}[generator finding; Step 4 in \Cref{alg:snf-hadamard-estimator}]\label{lem:generator_finding}
    Let $\gamma_{\rm gen} \in (0,1)$.
    A generator $g$ of the multiplicative group $\F_q^\times$ can be found with probability at least $1-\gamma_{\rm gen}$.
    Its classical bit complexity is $\widetilde O((\log q)^3 \log(1/\gamma_{\rm gen}))$ and quantum gate complexity is $\widetilde O((\log q)^3 \log(1/\gamma_{\rm gen}))$.
\end{lem}
\begin{proof}
    It is standard that $\F_q^\times$ is a cyclic group of order $q-1$. Hence
    $g\in\F_q^\times$ is a generator if and only if $\ord(g)=q-1$.
    
    We first compute the prime factorization  $q-1=\prod_{t=1}^T \ell_t^{e_t}$.
    This is done using Shor's quantum factoring algorithm \cite{shor1997PolynomialTime}.
    The factorization can be obtained with failure probability at
    most $\gamma_{\mathrm{gen}}/2$ with
        $\widetilde O\!\left(
            (\log q)^3\log(1/\gamma_{\mathrm{gen}})
        \right)$ gates.
    
    Now suppose that the factorization of $q-1$ has been obtained. We repeatedly
    sample $h$ from $\F_q^\times$ uniformly at random and test whether $h$
    is a generator. For each distinct prime divisor $\ell\mid q-1$, compute
     $h^{(q-1)/\ell}\in \F_q^\times$
    by repeated squaring. We accept $h$ if and only if
        $h^{(q-1)/\ell}\neq 1$
        for every prime  $\ell\mid q-1$.
    This criterion is correct because, in a cyclic group of order $q-1$, an
    element has order $q-1$ if and only if it does not lie in any proper subgroup
    of index $\ell$ for $\ell\mid q-1$.
    
    Let
    \[
        \omega(q-1)
        :=
        \#\{\ell:\ell \text{ prime and } \ell\mid q-1\}
    \]
    be the number of distinct prime divisors of $q-1$. Since
            $\omega(q-1)\le \log_2(q-1)=O(\log q)$,
    each trial requires at most $O(\log q)$ exponentiations in $\F_q$. Each
    exponentiation has exponent of bit length $O(\log q)$, and hence costs
    $O(\log q)$ multiplications in $\F_q$. Under the assumed arithmetic model,
    one multiplication in $\F_q$ costs $\widetilde O(\log q)$ bit operations.
    Therefore one trial costs
    \[
        O(\omega(q-1)\log q)\cdot \widetilde O(\log q)
        =
        \widetilde O((\log q)^3)
    \]
    bit operations.
    
    It remains to bound the number of trials. The proportion of generators in the
    cyclic group $\F_q^\times$ is
    \[
        p
        =
        \frac{\varphi(q-1)}{q-1}
    \]
    where $\varphi(\cdot)$ is Euler’s totient function.
    By the Rosser--Schoenfeld bound, there is an absolute constant $c>0$ such
    that, for all sufficiently large $q$,
    \[
        \frac{\varphi(q-1)}{q-1}
        \ge
        \frac{c}{\log\log(q-1)}.
    \]
    The finitely many remaining values of $q$ are absorbed into the implicit
    constant. Hence
    \[
        p\ge \frac{c}{\log\log q}
    \]
    after adjusting constants.
    
    If we perform $R$ independent trials, the probability that all trials fail is
            $(1-p)^R
        \le
        \exp(-pR)$.
    Thus it suffices to take
        $R
        =
        \left\lceil
            C\,\log\log q\cdot \log(2/\gamma_{\mathrm{gen}})
        \right\rceil$
    for a sufficiently large absolute constant $C>0$. With this choice,
    the sampling-and-testing stage fails with probability at most
    $\gamma_{\mathrm{gen}}/2$, and its bit complexity is
    \[
        R\cdot \widetilde O((\log q)^3)
        =
        \widetilde O\!\left(
            (\log q)^3\log(1/\gamma_{\mathrm{gen}})
        \right),
    \]
    where the factor $\log\log q$ is absorbed into the
    $\widetilde O$-notation.
    
    Combining this with the amplified Shor factorization step, the total failure
    probability is at most
$\gamma_{\mathrm{gen}}/2+\gamma_{\mathrm{gen}}/2
        =
        \gamma_{\mathrm{gen}}$,
    and the total bit and gate complexities are
        $\widetilde O\!\left(
            (\log q)^3\log(1/\gamma_{\mathrm{gen}})
        \right)$.
    This proves the claim.
\end{proof}

For the remaining primitives, fix a generator $g\in\F_q^\times$.  We use
this generator to define multiplicative character states.  For
$a\in \mathbb Z/(q-1)\mathbb Z$, set
\[
    \ket{\chi^a}
    :=
    \frac{1}{\sqrt{q-1}}
    \sum_{j=0}^{q-2}\omega_{q-1}^{aj}\ket{g^j},
    \qquad
    \omega_{q-1}:=e^{2\pi i/(q-1)}.
\]
The following procedure is the single-trial quantum subroutine used inside
\Cref{alg:snf-hadamard-estimator}.
\begin{algorithm}[H]
\caption{\textsc{HadamardTest}%
$(\F_q,g,k,z,\{\lambda_j\}_{k_j\ne0},
\eps_{\rm seed},\eps_{\rm ph},\gamma_\chi)$}
\makeatletter\protected@edef\@currentlabelname{\textsc{HadamardTest}}\makeatother
\label{alg:hadamard-test}
\begin{algorithmic}[1]
    \REQUIRE A finite field $\F_q$, a generator $g\in\F_q^\times$,
    a vector $k=(k_1,\ldots,k_s)\in K_U(q)$,
    $z=z(k)$, coefficient constants
    $\lambda_j=a_j^{-k_j}\in\F_q^\times$ for $k_j\ne0$,
    precisions $\eps_{\rm seed},\eps_{\rm ph}\in(0,1)$,
    and a state-preparation failure parameter $\gamma_\chi\in(0,1)$.

    \STATE Prepare a state $\ket{\widetilde\psi}$ satisfying
    $
    \left\|
    \ket{\widetilde\psi}
    -
    \ket{\chi^0}\ket{\chi^1}
    \right\|
    \le
    \eps_{\rm seed}
    $
    with failure probability at most $\gamma_\chi$ on the system register.

    \STATE Construct an $\eps_{\rm ph}$-approximate phase oracle
    $\widetilde U_k$ for the ideal unitary $U_k$ \eqref{eq:phase_oracle} satisfying
    $
    U_k\ket{\chi^0}\ket{\chi^1}
    =
    \Phi(k)\ket{\chi^0}\ket{\chi^1},
    $
    with input arguments $\F_q, g, k,$ and $\{\lambda_j\}$.
    (In the construction, each additive finite-field QFT is implemented to
    operator-norm error at most $\eps_{\rm ph}/s$.)

    \STATE Initialize one control qubit in $\ket{0}$.

    \STATE Apply a Hadamard gate to the control qubit.

    \STATE Apply controlled-$\widetilde U_k$ to the system register,
    controlled on the control qubit.

    \STATE Apply a Hadamard gate to the control qubit.

    \STATE Measure the control qubit in the computational basis and denote the
    outcome by $b\in\{0,1\}$.

    \STATE Set
    $
    X(k)= (-1)^b\in\{-1,+1\}.
    $

    \STATE \textbf{Return}
    $
    Y(k)= (-1)^z q^{-z/2} X(k).
    $
\end{algorithmic}
\end{algorithm}
The cost of \Cref{alg:hadamard-test} has two components.  The first is the
cost of preparing the input state
$\ket{\widetilde\psi}\approx \ket{\chi^0}\ket{\chi^1}$.  The second is the
cost of implementing the $k$-dependent phase oracle.  We first record the
state-preparation cost.

\begin{lem}[Cost of chi state preparation]\label{lem:chi-state-prep}
    Let $\eps_\mathrm{seed},\gamma_\chi\in(0,1)$.
    Suppose that a generator $g\in\F_q^\times$ is given.
    Assume that the cyclic QFT $F_{q-1}$ over $\Z/(q-1)\Z$ is implemented by a unitary $\widetilde F_{q-1}$ satisfying $\left\|\widetilde F_{q-1} - F_{q-1}\right\| \leq \frac{\eps_\mathrm{seed}}{4}$.
    Then one can prepare a state $\ket{\tilde \chi^1}$ satisfying $\left\| \ket{\tilde \chi^1} - \ket{\chi^1} \right\| \leq \eps_\mathrm{seed}$ with probability at least $1-\gamma_\chi$.
    This procedure has a gate complexity
    \begin{align}\label{eq:Tseed-def}
    T_\chi(q,\eps_{\rm seed},\gamma_\chi) = O\left(\frac{1+\log\log q}{(1-\eps_\mathrm{seed}/4)^2} \log\frac{1}{\gamma_\chi}\right) \cdot \widetilde O\left((\log q)^2 + (\log q + \log(1/\eps_\mathrm{seed}))^2\right).
    \end{align}
    Moreover, preparing a pair $\ket{\tilde\psi}=\ket{\chi^0}\ket{\tilde\chi^1}$ requires $T_\psi(q,\eps_{\rm seed},\gamma_\chi)=T_\chi(q,\eps_{\rm seed},\gamma_\chi)+O(\log q) = T_\chi(q,\eps_{\rm seed},\gamma_\chi)$ gates.
\end{lem}
\begin{proof}
    The procedure consists of independent attempts, repeated a fixed number $R$ of times.
    We follow the same procedure as \cite{van2003quantum,dam2004Quantum} for preparing $\ket{\chi^0}$ and $\ket{\chi^1}$, but here we analyze the costs in more detail, especially regarding the implementation error of QFT.
    
    The preparation of $\ket{\chi^0}$ proceeds by amplifying a standard uniform superposition into the target subspace representing $\F_q^\times$.
    Specifically, initialize the register $\ket{0}^{\otimes b_q}$ with $b_q = \lceil \log_2 q \rceil$ qubits and apply the Hadamard gate $H^{\otimes b_q}$ to generate the uniform superposition over all $2^{b_q}$ computational basis states.
    Then, $\ket{\chi^0}$ is deterministically obtained by exact amplitude amplification with the initial amplitude $a = \sqrt{(q-1)/2^{b_q}}$.
    Since the marking oracle $\ket{t}\ket{y}\ket{0^a} \mapsto \ket{t}\ket{y \oplus \mathbf{1}[\mathrm{enc}(t) \in G]}\ket{0^a}$, where $\mathrm{enc}(\cdot)$ is a bijective function mapping from a finite field element to a label in $\{0,...,q-1\}$ (assuming a zero element corresponds to $0$ label) and $G = \{1,...,q-1\}$, costs $O(b_q)$ Toffoli gates, the total gate complexity for preparing $\ket{\chi^0}$ is $O(\log q)$, incorporating a constant multiplicative overhead of $O(1/a) = O(1)$ from amplitude amplification. 
    
    We follow \cite[Algorithm 1]{van2003quantum} for preparing $\ket{\chi^1}$.
    First, we prepare $(1/\sqrt{q-1})\sum_{j=0}^{q-2}\ket{j}$.
    This can be done in the same way as $\ket{\chi^0}$ with a gate complexity of $O(\log q)$.
    Then we do modular exponentiation $\ket{j}\ket{0}\mapsto\ket{j}\ket{g^j}$ to obtain $\ket{\psi_1}=(1/\sqrt{q-1})\sum_{j=0}^{q-2}\ket{j}\ket{g^j}$ at a cost of $\widetilde O((\log q)^2)$ gates.
    The application of QFT over $\Z/(q-1)\Z$ (\Cref{prop:QFT_over_ZN}) to the first register yields the state $\ket{\psi_2}=(1/\sqrt{q-1})\sum_{k=0}^{q-2}\ket{k}\ket{\chi^k}$ using $O((\log q + \log(1/\eps))^2)$ gates.
    After measuring the first register, if the measurement outcome $k$ satisfies $\gcd(k,q-1)=1$, we proceed to the next step; otherwise we start over the first step.
    Finally, applying the map $P_k:\ket{x}\mapsto\ket{x^k}$, which costs $\widetilde O((\log q)^2)$ gates, we obtain the transformation $P_k\ket{\chi^k}=\ket{\chi^1}$, yielding the desired state.
    
    Let $\eps_{\rm qft}$ denote the operator-norm error of the cyclic QFT
implementation.
    Now let us assume that QFT over $\Z/(q-1)\Z$ can be implemented only approximately such that $\|\widetilde F_{q-1} - F_{q-1}\| \leq \eps_{\rm qft}$.
    We introduce
    \begin{equation}
        \ket{\tilde \psi_2} = (\widetilde F_{q-1} \otimes I) \ket{\psi_1} = \ket{\psi_2} + \ket{\Delta}, \quad \ket{\Delta} = \frac{1}{\sqrt{q-1}} \sum_{k=0}^{q-2} \sum_{j=0}^{q-2} E_{kj} \ket{k} \ket{g^j} = \sum_{k=0}^{q-2} \ket{k} \ket{\Delta_k}.
    \end{equation}
    Then, we have
    \begin{equation}
        \|\ket{\Delta_k}\|^2 = \frac{1}{q-1} \sum_{j=0}^{q-2} |E_{kj}|^2 = \frac{1}{q-1} \|(\widetilde F_{q-1} - F_{q-1})^\dagger \ket{k}\|^2 \leq \frac{1}{q-1} \|(\widetilde F_{q-1} - F_{q-1})^\dagger \|^2 \leq \frac{\eps_{\mathrm{qft}}^2}{q-1},
    \end{equation}
    For the unnormalized post-measurement state $v_k = \ket{k}((1/\sqrt{q-1})\ket{\chi^k}+\ket{\Delta_k})$ conditioned on observing $k$ in the first register, we have
    \begin{equation}
        \left\| \frac{v_k}{\|v_k\|} - \ket{k}\ket{\chi^k} \right\| \leq \frac{2\eps_{\rm qft}}{1-\eps_{\rm qft}}.
    \end{equation}
    The probability of measuring $k$, given by $\|v_k\|^2$, satisfies $(1-\eps_{\rm qft})^2/(q-1) \leq \|v_k\|^2 \leq (1+\eps_{\rm qft})^2/(q-1)$.
    Thus, the probability of measuring $k$ such that $\gcd(k,q-1)=1$ is at least \[(\varphi(q-1)/(q-1)) \cdot (1-\eps_{\rm qft})^2 \geq \frac{c (1-\eps_{\rm qft})^2}{1+\log\log(q+1)}\] for a constant $c>0$.
    Running \[R\defeq\left\lceil\frac{1+\log\log(q+1)}{c (1-\eps_{\rm qft})^2}\log\frac{1}{\gamma_\chi}\right\rceil\] independent attempts guarantees that the failure probability is bounded by $\gamma_\chi$.
    For one attempt, inserting $\eps_{\rm qft}=\eps_\mathrm{seed}/4$ to satisfy the desired precision, the gate complexity of obtaining $\ket{\chi^1}$ with state error in $l_2$-norm at most $\eps_\mathrm{seed}$ is
    \begin{equation}
        \widetilde O\left((\log q)^2 + (\log q + \log(1/\eps_\mathrm{seed}))^2\right).
    \end{equation}
    The final complexity is obtained by multiplying this by $R$.
\end{proof}
We next combine the prepared chi states with the $k$-dependent phase oracle.
The ideal phase oracle has $\ket{\chi^0}\ket{\chi^1}$ as an eigenvector, and
the corresponding eigenvalue is precisely $\Phi(k)$.  The following theorem
bounds both the bias caused by approximate state preparation and approximate
finite-field QFTs, and the gate complexity of one Hadamard-test trial.
\begin{lem}[Quantum per-sample cost; Step 10 in \Cref{alg:snf-hadamard-estimator}]\label{lem:quantum-per-sample}
    Let $\eps_{\rm seed},\eps_{\rm ph},\gamma_\chi\in(0,1)$.
    Assume that a generator $g\in\F_q^\times$ is given, and that each
    additive finite-field QFT $F_{\F_q}$ used in the phase oracle is
    implemented by a unitary $\widetilde F_{\F_q}$ with
    $
        \|\widetilde F_{\F_q}-F_{\F_q}\|\le \eps_{\rm ph}/s .
    $
    Then \Cref{alg:hadamard-test} returns $Y(k)\in[-1,1]$ such that,
    conditioned on successful state preparation $\mathcal S_\chi$,
    \[
        \left|
        \mathbb E[Y(k)\mid k,\mathcal S_\chi]
        -
        \operatorname{Re}T(k)
        \right|
        \le
        q^{-z(k)/2}(\eps_{\rm ph}+2\eps_{\rm seed})
        \le
        \eps_{\rm ph}+2\eps_{\rm seed}.
    \]
    In particular, if
    $
        \eps_{\rm ph}+2\eps_{\rm seed}\le \eta_\Xi/2,
    $
    then
    \[
        \left|
        \mathbb E[Y(k)\mid k,\mathcal S_\chi]
        -
        \operatorname{Re}T(k)
        \right|
        \le \eta_\Xi/2.
    \]
    The gate complexity is
    \[
        C_{\rm samp}^{\rm q}
        =
        T_\psi(q,\eps_{\rm seed},\gamma_\chi)
        +
        \widetilde O\!\left(
            s\left(
                (\log q)^2
                +
                \log q\,(\log q+\log(s/\eps_{\rm ph}))^2
            \right)
        \right),
    \]
    where the first term is the cost for preparing $\ket{\tilde\psi}$ from \Cref{lem:chi-state-prep}.
\end{lem}
\begin{proof}
    Although we follow a similar procedure as \cite[Algorithm 1]{dam2004Quantum},
    we use one Hadamard test shot only to produce a random variable $X(k)$ such that 
    \[
        \mathbb E[X(k)\mid k]
        =
        \operatorname{Re}
        \bra{\chi^0}\bra{\chi^1}
        U_k
        \ket{\chi^0}\ket{\chi^1}
        =
        \operatorname{Re}\Phi(k),
    \]
    in the ideal case.
    For the input state $\ket{\psi}$, each trial prepares $\ket{\chi^0}$ and $\ket{\tilde\chi^1}$.
    By \Cref{lem:chi-state-prep}, one can produce $\ket{\tilde\chi^1}$ satisfying $\|\ket{\tilde\chi^1}-\ket{\chi^1}\|\leq\eps_\mathrm{seed}$.
    Thus the input state $\ket{\tilde\psi}=\ket{\chi^0}\ket{\tilde\chi^1}$ satisfies $\|\ket{\tilde\psi}-\ket{\psi}\| \leq \eps_\mathrm{seed}$.
    Now, write $\lambda_j=\lambda_j(k):=a_j^{-k_j}\in\F_q^\times$ for $k_j\ne0$.
    For the phase oracle 
    \begin{equation}\label{eq:phase_oracle}
        U_k = \left(\prod_{\substack{1\le j\le s\\ k_j\ne 0}} (I \otimes U_{mul}^{(j)}) \right) \left(\prod_{\substack{1\le j\le s\\ k_j\ne 0}} D^{-k_j} (F_{\F_{q}} \otimes I) D^{-k_j} \right),
    \end{equation}
    with the modular multiplication $U_{mul}^{(j)}:\ket{y}\mapsto\ket{y \cdot \lambda_j}$ using $\widetilde O((\log q)^2)$ gates, a division operator $D^\alpha:\ket{x}\ket{y}\mapsto\ket{x}\ket{y/x^\alpha}$ using $\widetilde O((\log q)^2)$ gates, and the additive finite-field QFT $F_{\F_q}$ in \Cref{prop:finite_field_QFT} using $O(r(\log p + \log(r/\eps))^2)$, bounded by $O(\log q(\log q + \log(1/\eps))^2)$ gates.
    Noticing that 
    \begin{equation}
        D^{-k_j} (F_{\F_{q}} \otimes I) D^{-k_j} \ket{\chi^0} \ket{\chi^1} = D^{-k_j} (F_{\F_{q}} \otimes I) \ket{\chi^{-k_j}} \ket{\chi^1}
        = D^{-k_j} \frac{G(\chi^{-k_j})}{\sqrt{q}} \ket{\chi^{k_j}} \ket{\chi^1}
        = \frac{G(\chi^{-k_j})}{\sqrt{q}} \ket{\chi^0} \ket{\chi^1},
    \end{equation}
    and
    \begin{equation}
        U_{\mathrm{mul}}^{(j)} \ket{\chi^1} = \chi(\lambda_j^{-1})\ket{\chi^1}
    =
    \chi(a_j^{k_j})\ket{\chi^1}
    =
    \chi^{k_j}(a_j)\ket{\chi^1},
    \end{equation}
    we have $U_k\ket{\chi^0}\ket{\chi^1}=\Phi(k)\ket{\chi^0}\ket{\chi^1}$.
    Substituting $F_{\F_q}$ in the phase oracle $U_k$ with $\tilde F_{\F_q}$, we can implement the approximate phase oracle $\tilde U_k$ satisfying $\|\tilde U_k - U_k\| \leq (s - z(k)) (\eps_\mathrm{ph}/s) \leq \eps_\mathrm{ph}$.
    Then we have 
    \begin{align}
        \abs{\mathbb E[\widetilde X(k)\mid k,\mathcal S_\chi]-\operatorname{Re}\Phi(k)} &= \left| \operatorname{Re} \braket{\tilde\psi|\tilde U_k|\tilde\psi} - \operatorname{Re} \braket{\psi|U_k|\psi} \right| \\
        &\leq \left|\braket{\tilde\psi|\tilde U_k|\tilde\psi}-\braket{\psi|U_k|\psi}\right| \\
        &\leq \left|\braket{\tilde\psi|(\tilde U_k-U_k)|\tilde\psi}\right| + \left|\braket{\tilde\psi|U_k|\tilde\psi}-\braket{\psi|U_k|\psi}\right| \\
        &\leq \eps_\mathrm{ph} + 2 \|\ket{\tilde\psi}-\ket{\psi}\| \\
        &\leq \eps_\mathrm{ph} + 2 \eps_\mathrm{seed}.
    \end{align}
    Since
    $
        Y(k)=(-1)^{z(k)}q^{-z(k)/2}\widetilde X(k)
    $ and $
        T(k)=(-1)^{z(k)}q^{-z(k)/2}\Phi(k),
    $
    we obtain
    \[
        \left|
        \mathbb E[Y(k)\mid k,\mathcal S_\chi]
        -
        \operatorname{Re}T(k)
        \right|
        \le
        q^{-z(k)/2}(\eps_{\rm ph}+2\eps_{\rm seed})
        \le
        \eps_{\rm ph}+2\eps_{\rm seed},
    \]
    completing the stated precision.
    Finally, the gate complexity is the sum of the chi state preparation cost and the cost of implementing the phase oracle.
    The arithmetic part contributes $\widetilde O(s\cdot(\log q)^2)$, while the finite-field QFTs contribute $\widetilde O\left( s \log q(\log q + \log(s/\eps_\mathrm{ph}))^2\right)$.
    The Hadamard test uses a controlled version of $\widetilde U_k$. 
    This changes the gate count only by a constant-factor overhead, maintaining the same implementation error $\eps_{\rm ph}$.
    Together with \Cref{lem:chi-state-prep}, this gives the stated per-sample complexity.
\end{proof}
We are now ready to assemble the preceding estimates and prove the primitive
cost statement used in the main text.

\noindent
\textbf{Proof of \Cref{prop:implemented-primitives}.}
We match the four items in \Cref{prop:implemented-primitives} with the
estimates proved above.

For item~(1), \Cref{lem:classical-preprocessing-cost} gives
\[
T_{\rm pre}^{\rm cl}
=
\widetilde O\!\left(
m^6L+s\log q+n\log q+\log(1/\eps)+\log(1/\delta_{\rm MC})
\right).
\]
Since $m=\max\{n+1,s\}$ and $L=\lceil\log_2(1+q+H)\rceil$, the lower-order
terms are absorbed into the displayed $\widetilde O(m^6L)$ preprocessing
bound in the statement of \Cref{prop:implemented-primitives}.

For item~(2), \Cref{lem:classical-per-sample} gives the expected classical
per-sample controller cost
\[
C_{\rm samp}^{\rm cl}
=
\widetilde O\!\left(
s^2\log q+s(\log q)^2+\log M
\right).
\]
The term $\log M$ accounts only for the final accumulator update in the
Monte Carlo loop.  The cost of sampling $k\in K_U(q)$ and computing the
$k$-dependent coefficient constants is therefore
\[
\widetilde O\!\left(
s^2\log q+s(\log q)^2
\right),
\]
as claimed.

For item~(3), \Cref{lem:generator_finding} gives a generator of $\F_q^\times$ with failure
probability at most $\gamma_{\rm gen}$ and quantum gate complexity
\[
\widetilde O\!\left((\log q)^3\log(1/\gamma_{\rm gen})\right).
\]
The additional classical verification and field-arithmetic work in the proof
is of the same polylogarithmic order and is included in the hybrid
bookkeeping.

It remains to prove item~(4).  Apply \Cref{alg:hadamard-test} with
\[
\eps_{\rm ph}=\eps_{\rm seed}=\eta_\Xi/6.
\]
Then
\[
\eps_{\rm ph}+2\eps_{\rm seed}
=
\eta_\Xi/2.
\]
By \Cref{lem:quantum-per-sample}, conditioned on successful preparation of
$\ket{\widetilde\psi}$, the Hadamard-test outcome $X(k)$ satisfies
\[
\left|
\mathbb E[X(k)\mid k]
-
\operatorname{Re}\Phi(k)
\right|
\le
\eps_{\rm ph}+2\eps_{\rm seed}
=
\eta_\Xi/2.
\]
Since \Cref{alg:hadamard-test} returns
\[
Y(k)=(-1)^{z(k)}q^{-z(k)/2}X(k),
\]
whereas
\[
T(k)=(-1)^{z(k)}q^{-z(k)/2}\Phi(k),
\]
and $q^{-z(k)/2}\le 1$, we obtain
\[
\left|
\mathbb E[Y(k)\mid k,\mathcal S_\chi]
-
\operatorname{Re}T(k)
\right|
\le
\eta_\Xi/2.
\]
This is the bias bound asserted in item~(4).

Finally, substituting
$\eps_{\rm ph}=\eps_{\rm seed}=\eta_\Xi/6$ and
$\gamma_\chi=\gamma$ into \Cref{lem:quantum-per-sample} gives
\[
C_{\rm samp}^{\rm q}
=
\widetilde O\!\left(
(\log q+\log(1/\eta_\Xi))^2
(\log(1/\eta_\Xi)+\log(1/\gamma))
+
s\log q(\log q+\log s+\log(1/\eta_\Xi))^2
\right).
\]
Here the first term comes from chi-state preparation, and the second from the
$k$-dependent phase oracle and its controlled use in the Hadamard test.
This proves item~(4), and hence completes the proof of
\Cref{prop:implemented-primitives}.
\qed

\section{Hardness for approximate counting for ordinary polynomials}
\label[appendix]{app:Similar results for ordinary polynomials}

We consider an analog of \Cref{prob:torus-approx} for ordinary polynomials.
For $f\in\F_q[x_1,\ldots,x_n]$, let us define $U$ and $\rho$ in a similar way.
We write 
\[
  N(f):=\#\{z\in\F_q^n:f(z)=0\}.
\]
\begin{prob}
\label{prob:rank_sensitive_approx_en}
Given a finite field $\F_q$, integers $n,s\ge 1$, a list $\bigl((u_1,a_1),\ldots,(u_s,a_s)\bigr)$ where
$u_j\in\Z_{\geq 0}^n$ and $a_j\in\F_q^\times$, which defines an ordinary polynomial $f$, and  parameters $0<\eps,\delta<1$,
 output a number $\widehat N(f)\in\Q$ such that
\[
\Pr\!\left[
  \left|\widehat N(f)-N(f)\right|
  \le \eps q^{n+s/2-\rho}
\right]
\ge 1-\delta.
\]
\end{prob}

For a finite poset $P$, an \emph{antichain} is a subset whose distinct elements are pairwise incomparable.
\begin{thm}[Antichain counting]
\label{thm:antichain_hard_en}
Counting antichains of a finite poset is \SharpP-complete.
\end{thm}
This is a standard consequence of the counting-complexity results of Provan--Ball \cite{ProvanBall1983}.  Cattani--Dickenstein also use poset antichain counting as the basic hardness source in their \SharpP-completeness result for counting solutions of binomial complete intersections \cite[Section~4 and Theorem~4.3]{CattaniDickenstein2007}.

Below, we show that the approximation problem for ordinary polynomials is also $\#$P-hard under some reductions.

\begin{thm}
\label{thm:affine_hardness_en}
For every fixed $0<\eps<1/4$, \Cref{prob:rank_sensitive_approx_en}
with $\delta=1/3$ is $\SharpP$-hard under randomized polynomial-time
Turing reductions. 
\end{thm}

\begin{proof}
A subset $I\subseteq P$ of a finite poset is an order ideal if $a\in I$ and $b\le a$ imply $b\in I$.  For finite posets, order ideals are in bijection with antichains: an order ideal maps to its set of maximal elements, and an antichain $A$ maps to the order ideal generated by $A$.  Thus order-ideal counting and antichain counting are equivalent.

We reduce from antichain counting.  Let $P=\{1,\ldots,R\}$ be a finite poset, numbered so that $b\prec a$ implies $b<a$.  For each $a\in P$, introduce a variable $x_a$ and define
\[
  p_a(x):=x_a-x_a^2\prod_{b\prec a}x_b,
\]
with the empty product interpreted as $1$.

The common zeros of the system $p_a(x)=0$ are in bijection with the order ideals of $P$.  Indeed, by induction in the chosen order, every solution is $0$-$1$ valued.  Moreover, if $x_a=1$, then all predecessors $b\prec a$ must satisfy $x_b=1$, so the set $\{a:x_a=1\}$ is an order ideal.  Conversely, the indicator vector of any order ideal satisfies all equations.

Introduce auxiliary variables $y_1,\ldots,y_R$ and set
\[
  H_P(x,y):=
  \sum_{a=1}^R y_a p_a(x)
  =
  \sum_{a=1}^R
  \left(y_ax_a-y_ax_a^2\prod_{b\prec a}x_b\right).
\]
This is a polynomial in $2R$ variables.  Each $a$ contributes two monomials, and different values of $a$ have different $y_a$-factors, so there are no collisions; hence $s=2R$.

For a fixed $x\in\F_q^R$, if $p_1(x)=\cdots=p_R(x)=0$, then $H_P(x,y)=0$ for all $y\in\F_q^R$.  Otherwise $H_P(x,y)=0$ is a nonzero linear equation in the $y$ variables and has $q^{R-1}$ solutions.  If $A(P)$ denotes the number of order ideals of $P$, then
\[
  N(H_P)=A(P)q^R+(q^R-A(P))q^{R-1}
  =q^{2R-1}+A(P)(q-1)q^{R-1}.
\]

We next compute the support rank.  For each $a$, let $u_a^{(0)}$ and $u_a^{(1)}$ be the two exponent columns corresponding to the two monomials in $y_ap_a(x)$.  Their difference $d_a:=u_a^{(1)}-u_a^{(0)}$
has no $y$-coordinates and has $x$-coordinates $e_a+\sum_{b\prec a}e_b$.
The chosen ordering makes these vectors lower triangular with diagonal entries $1$, hence they are linearly independent.  The columns $u_1^{(0)},\ldots,u_R^{(0)}$ have distinct standard-basis $y$-coordinates and are independent from the $d_a$.  Therefore the support columns have rank $2R$.  The augmented top row is the sum of the $y$-rows, because every monomial contains exactly one $y_a$.  Thus $n=s=\rho=2R$.

Now set $q=3$.  The allowed additive error is
$\eps q^{n+s/2-\rho}=\eps q^R$.
Increasing $A(P)$ by $1$ changes $N(H_P)$ by
$\Delta=(q-1)q^{R-1}$.
Since $0<\eps<1/4$ and $q=3$, we have $\eps q^R<\frac12(q-1)q^{R-1}$.
Thus an approximation $\widetilde N$ determines
  \[\frac{\widetilde N-q^{2R-1}}{(q-1)q^{R-1}}\]
within distance $<1/2$ of the integer $A(P)$, so rounding recovers
$A(P)$ exactly whenever the approximation oracle succeeds.  Since the
oracle is invoked with $\delta=1/3$, the randomized reduction succeeds
with probability at least $2/3$.  By \Cref{thm:antichain_hard_en}, this
proves $\SharpP$-hardness.
\end{proof}

\bibliographystyle{plainnat}
\bibliography{main}

@article{SperberVoight2013,
  title={Computing zeta functions of nondegenerate hypersurfaces with few monomials},
  author={Sperber, Steven and Voight, John},
  journal={LMS Journal of Computation and Mathematics},
  volume={16},
  pages={9--44},
  year={2013},
  publisher={London Mathematical Society}
}

@article{AdolphsonSperber1989,
  author  = {Adolphson, Alan and Sperber, Steven},
  title   = {Exponential sums and {Newton} polyhedra: cohomology and estimates},
  journal = {Annals of Mathematics},
  year    = {1989},
  volume  = {130},
  number  = {2},
  pages   = {367--406}
}

@article{Schoof1985,
  title={Elliptic curves over finite fields and the computation of square roots mod $p$},
  author={Schoof, Ren{\'e}},
  journal={Mathematics of computation},
  volume={44},
  number={170},
  pages={483--494},
  year={1985}
}

@article{Kedlaya2001,
  author  = {Kedlaya, Kiran S.},
  title   = {Counting points on hyperelliptic curves using {Monsky--Washnitzer} cohomology},
  journal = {Journal of the Ramanujan Mathematical Society},
  year    = {2001},
  volume  = {16},
  number  = {4},
  pages   = {323--338}
}

@incollection{LauderWan2008,
  author    = {Lauder, Alan G. B. and Wan, Daqing},
  title     = {Counting points on varieties over finite fields of small characteristic},
  booktitle = {Algorithmic Number Theory: Lattices, Number Fields, Curves and Cryptography},
  editor    = {Buhler, J. P. and Stevenhagen, P.},
  series    = {MSRI Publications},
  volume    = {44},
  pages     = {579--612},
  publisher = {Cambridge University Press},
  year      = {2008}
}

@article{Newman1971SNF,
  author  = {Newman, Morris},
  title   = {On the {S}mith normal form},
  journal = {Journal of Research of the National Bureau of Standards, Section B},
  year    = {1971},
  volume  = {75B},
  number  = {1--2},
  pages   = {81--84}
}

@inproceedings{Storjohann1996,
  title={Near optimal algorithms for computing Smith normal forms of integer matrices},
  author={Storjohann, Arne},
  booktitle={Proceedings of the 1996 international symposium on Symbolic and algebraic computation},
  pages={267--274},
  year={1996}
}

@book{LidlNiederreiter1997,
  author    = {Rudolf Lidl and Harald Niederreiter},
  title     = {Finite Fields},
  edition   = {2nd},
  publisher = {Cambridge University Press},
  year      = {1997}
}

@book{BerndtEvansWilliams1998,
 author = {Berndt, Bruce C. and Evans, Ronald J. and Williams, Kenneth S.},
 title = {Gauss and {Jacobi} sums},
 fseries = {Canadian Mathematical Society Series of Monographs and Advanced Texts},
 series = {Can. Math. Soc. Ser. Monogr. Adv. Texts},
 year = {1998},
 publisher = {New York, NY: John Wiley \& Sons},
}

@inproceedings{Williams2018Counting,
  author    = {Williams, R. Ryan},
  title     = {Counting Solutions to Polynomial Systems via Reductions},
  booktitle = {1st Symposium on Simplicity in Algorithms ({SOSA} 2018)},
  series    = {OASIcs},
  volume    = {61},
  pages     = {6:1--6:15},
  year      = {2018},
  publisher = {Schloss Dagstuhl--Leibniz-Zentrum f{\"u}r Informatik}
}

@article{LangWeil1954,
  author  = {Lang, Serge and Weil, Andr{\'e}},
  title   = {Number of points of varieties in finite fields},
  journal = {American Journal of Mathematics},
  year    = {1954},
  volume  = {76},
  number  = {4},
  pages   = {819--827}
}

@article{CafureMatera2006,
  author  = {Cafure, Antonio and Matera, Guillermo},
  title   = {Improved explicit estimates on the number of solutions of equations over a finite field},
  journal = {Finite Fields and Their Applications},
  year    = {2006},
  volume  = {12},
  number  = {2},
  pages   = {155--185}
}

@incollection{GhorpadeLachaud2002,
  title={Number of solutions of equations over finite fields and a conjecture of {L}ang and {W}eil},
  author={Ghorpade, Sudhir R and Lachaud, Gilles},
  booktitle={Number theory and discrete mathematics},
  pages={269--291},
  year={2002},
  publisher={Springer}
}

@inproceedings{vonZurGathenKarpinskiShparlinski1997,
  title={Counting curves and their projections},
  author={Von Zur Gathen, Joachim and Karpinski, Marek and Shparlinski, Igor},
  booktitle={Proceedings of the twenty-fifth annual ACM symposium on Theory of Computing (STOC)},
  pages={805--812},
  year={1993}
}

@inproceedings{lokshtanov2017beating,
  title={Beating brute force for systems of polynomial equations over finite fields},
  author={Lokshtanov, Daniel and Paturi, Ramamohan and Tamaki, Suguru and Williams, Ryan and Yu, Huacheng},
  booktitle={Proceedings of the Twenty-Eighth {A}nnual {ACM-SIAM} {S}ymposium on {D}iscrete {A}lgorithms (SODA)},
  pages={2190--2202},
  year={2017},
  organization={SIAM}
}

@article{schoof1995counting,
  title={Counting points on elliptic curves over finite fields},
  author={Schoof, Ren{\'e}},
  journal={Journal de th{\'e}orie des nombres de Bordeaux},
  volume={7},
  number={1},
  pages={219--254},
  year={1995}
}

@article{wan2008algorithmic,
  title={Algorithmic theory of zeta functions over finite fields},
  author={Wan, Daqing},
  journal={Algorithmic number theory: lattices, number fields, curves and cryptography},
  volume={44},
  pages={551--578},
  year={2008},
  publisher={Cambridge University Press Cambridge}
}

@phdthesis{KrawitzThesis,
  author = {Krawitz, Michael},
  title  = {FJRW Rings and Landau--Ginzburg Mirror Symmetry},
  school = {University of Michigan},
  year   = {2010},
  url    = {https://deepblue.lib.umich.edu/handle/2027.42/77910}
}

@article{AdolphsonSperber1990,
  author  = {Adolphson, Alan and Sperber, Steven},
  title   = {Exponential Sums on {$(\mathbb{G}_m)^n$}},
  journal = {Inventiones Mathematicae},
  volume  = {101},
  pages   = {63--79},
  year    = {1990}
}

@article{JerrumSinclairVigoda2004,
  title={A polynomial-time approximation algorithm for the permanent of a matrix with nonnegative entries},
  author={Jerrum, Mark and Sinclair, Alistair and Vigoda, Eric},
  journal={Journal of the {ACM} ({JACM})},
  volume={51},
  number={4},
  pages={671--697},
  year={2004},
  publisher={ACM New York, NY, USA}
}

@article{karpinski1993approximating,
  title={Approximating the number of zeroes of a {GF}[2] polynomial},
  author={Karpinski, Marek and Luby, Michael},
  journal={Journal of Algorithms},
  volume={14},
  number={2},
  pages={280--287},
  year={1993},
  publisher={Elsevier}
}

@inproceedings{grigoriev1991approximation,
  title     = {An Approximation Algorithm for the Number of Zeros of Arbitrary Polynomials over {GF}[$q$]},
  author    = {Grigoriev, Dima and Karpinski, Marek},
  booktitle = {Proceedings of the 32nd Annual Symposium on Foundations of Computer Science ({FOCS})},
  pages     = {662--669},
  year      = {1991},
  publisher = {IEEE}
}

@inproceedings{dell2025solving,
  title={Solving polynomial equations over finite fields},
  author={Dell, Holger and Haak, Anselm and Kallmayer, Melvin and Wennmann, Leo},
  booktitle={Proceedings of the 2025 Annual ACM-SIAM Symposium on Discrete Algorithms (SODA)},
  pages={2779--2803},
  year={2025},
  organization={SIAM}
}

@inproceedings{bjorklund2019solving,
  title={Solving systems of polynomial equations over $\mathrm{GF}(2)$ by a parity-counting self-reduction},
  author={Bj{\"o}rklund, Andreas and Kaski, Petteri and Williams, Ryan},
  booktitle={International Colloquium on Automata, Languages, and Programming},
  pages={1--13},
  year={2019},
  organization={Schloss Dagstuhl-Leibniz-Zentrum f{\"u}r Informatik}
}

@inproceedings{dinur2021improved,
  title={Improved algorithms for solving polynomial systems over $\mathrm{GF}(2)$ by multiple parity-counting},
  author={Dinur, Itai},
  booktitle={Proceedings of the 2021 {ACM-SIAM} {S}ymposium on {D}iscrete {A}lgorithms ({SODA})},
  pages={2550--2564},
  year={2021},
  organization={SIAM}
}

@inproceedings{huang1996solving,
  title={Solving systems of polynomial congruences modulo a large prime},
  author={Huang, Ming-Deh and Wong, Yiu-Chung},
  booktitle={Proceedings of 37th {C}onference on {F}oundations of {C}omputer {S}cience (FOCS)},
  pages={115--124},
  year={1996},
  organization={IEEE}
}

@book{vonZurGathenGerhard2013,
  author    = {Joachim von zur Gathen and J{\"u}rgen Gerhard},
  title     = {Modern Computer Algebra},
  edition   = {3},
  publisher = {Cambridge University Press},
  year      = {2013}
}

@article{Stanley2016SNF,
  author  = {Stanley, Richard P.},
  title   = {Smith Normal Form in Combinatorics},
  journal = {Journal of Combinatorial Theory, Series A},
  volume  = {144},
  pages   = {476--495},
  year    = {2016}
}

@techreport{karpinski1991lhotzky,
  author      = {Karpinski, Marek and Lhotzky, Barbara},
  title       = {An {$(\epsilon,\delta)$}-Approximation Algorithm for the Number of Zeros for a Multilinear Polynomial over {GF}[$q$]},
  institution = {International Computer Science Institute},
  number      = {TR-91-022},
  address     = {Berkeley, CA},
  year        = {1991}
}

@article{Valiant1979Permanent,
  title={The complexity of computing the permanent},
  author={Valiant, Leslie G},
  journal={Theoretical computer science},
  volume={8},
  number={2},
  pages={189--201},
  year={1979},
  publisher={Elsevier}
}

@article{ProvanBall1983,
  author  = {Provan, J. Scott and Ball, Michael O.},
  title   = {The Complexity of Counting Cuts and of Computing the Probability that a Graph is Connected},
  journal = {SIAM Journal on Computing},
  volume  = {12},
  number  = {4},
  pages   = {777--788},
  year    = {1983}
}

@article{CattaniDickenstein2007,
  author  = {Cattani, Eduardo and Dickenstein, Alicia},
  title   = {Counting Solutions to Binomial Complete Intersections},
  journal = {Journal of Complexity},
  volume  = {23},
  number  = {1},
  pages   = {82--107},
  year    = {2007},
  eprint  = {math/0510520},
  archivePrefix = {arXiv}
}

@article{ChengHillWan2013,
  title={Counting value sets: algorithm and complexity},
  author={Cheng, Qi and Hill, Joshua and Wan, Daqing},
  journal={The Open Book Series},
  volume={1},
  number={1},
  pages={235--248},
  year={2013},
  publisher={Mathematical Sciences Publishers}
}

@article{Milovanov2019,
  title={\#{P}-completeness of counting roots of a sparse polynomial},
  author={Milovanov, Alexey},
  journal={Information Processing Letters},
  volume={142},
  pages={77--79},
  year={2019},
  publisher={Elsevier}
}

@article{Slavov2023,
  title={Nearly sharp {L}ang--{W}eil bounds for a hypersurface},
  author={Slavov, Kaloyan},
  journal={Canadian Mathematical Bulletin},
  volume={66},
  number={2},
  pages={654--664},
  year={2023},
  publisher={Canadian Mathematical Society}
}

@article{Kedlaya2006,
  title={Quantum computation of zeta functions of curves},
  author={Kedlaya, Kiran S},
  journal={computational complexity},
  volume={15},
  number={1},
  pages={1--19},
  year={2006},
  publisher={Springer}
}

@article{Harvey2015,
  title={Computing zeta functions of arithmetic schemes},
  author={Harvey, David},
  journal={Proceedings of the London Mathematical Society},
  volume={111},
  number={6},
  pages={1379--1401},
  year={2015},
  publisher={Oxford University Press}
}

@inproceedings{HuangWong1998,
  title={An algorithm for approximate counting of points on algebraic sets over finite fields},
  author={Huang, Ming-Deh and Wong, Yiu-Chung},
  booktitle={International Algorithmic Number Theory Symposium},
  pages={514--527},
  year={1998},
  organization={Springer}
}

@article{CostaHarveyKedlaya2019,
  title={Zeta functions of nondegenerate hypersurfaces in toric varieties via controlled reduction in p-adic cohomology},
  author={Costa, Edgar and Harvey, David and Kedlaya, Kiran},
  journal={The Open Book Series},
  volume={2},
  number={1},
  pages={221--238},
  year={2019},
  publisher={Mathematical Sciences Publishers}
}

@article{Pila1990,
  author  = {Pila, Jonathan},
  title   = {Frobenius maps of abelian varieties and finding roots of unity in finite fields},
  journal = {Mathematics of Computation},
  volume  = {55},
  number  = {192},
  pages   = {745--763},
  year    = {1990}
}

@article{shor1997PolynomialTime,
  title={Polynomial-time algorithms for prime factorization and discrete logarithms on a quantum computer},
  author={Shor, Peter W},
  journal={SIAM review},
  volume={41},
  number={2},
  pages={303--332},
  year={1999},
  publisher={SIAM}
}

@misc{van2003quantum,
  title={Quantum algorithms for estimating {G}auss sums and calculating discrete logarithms},
  author={van Dam, Wim and Seroussi, Gadiel},
  year={2003},
  url = {https://sites.cs.ucsb.edu/~vandam/gausssumdlog.pdf}
}

@article{dam2004Quantum,
  title={Quantum computing and zeroes of zeta functions},
  author={van Dam, Wim},
  journal={arXiv preprint quant-ph/0405081},
  year={2004}
}

@inproceedings{hales2000Improved,
  title = {An Improved Quantum {{Fourier}} Transform Algorithm and Applications},
  booktitle = {Proceedings 41st {{Annual Symposium}} on {{Foundations}} of {{Computer Science}}},
  author = {Hales, L. and Hallgren, S.},
  year = 2000,
  pages = {515--525},
  issn = {0272-5428}
}

@article{harvey2021Integer,
  title={Integer multiplication in time O($n$log $n$)},
  author={Harvey, David and Van Der Hoeven, Joris},
  journal={Annals of Mathematics},
  volume={193},
  number={2},
  pages={563--617},
  year={2021},
  publisher={Department of Mathematics, Princeton University Princeton, New Jersey, USA}
}

@article{KarpLubyMadras1989,
  title={Monte-Carlo approximation algorithms for enumeration problems},
  author={Karp, Richard M and Luby, Michael and Madras, Neal},
  journal={Journal of algorithms},
  volume={10},
  number={3},
  pages={429--448},
  year={1989},
  publisher={Elsevier}
}

@article{DyerGoldbergGreenhillJerrum2004,
  title={The relative complexity of approximate counting problems},
  author={Dyer, Martin and Goldberg, Leslie Ann and Greenhill, Catherine and Jerrum, Mark},
  journal={Algorithmica},
  volume={38},
  number={3},
  pages={471--500},
  year={2004},
  publisher={Springer}
}

@article{DyerGoldbergJerrum2010,
  title={An approximation trichotomy for Boolean \#{CSP}},
  author={Dyer, Martin and Goldberg, Leslie Ann and Jerrum, Mark},
  journal={Journal of Computer and System Sciences},
  volume={76},
  number={3-4},
  pages={267--277},
  year={2010},
  publisher={Elsevier}
}

@article{GoldbergGuo2017ComplexIsingTutte,
  title={The complexity of approximating complex-valued {I}sing and {T}utte partition functions},
  author={Goldberg, Leslie Ann and Guo, Heng},
  journal={computational complexity},
  volume={26},
  number={4},
  pages={765--833},
  year={2017},
  publisher={Springer}
}

@article{GalanisGoldbergHerreraPoyatos2022,
  title={The complexity of computing the sign of the {T}utte polynomial},
  author={Goldberg, Leslie Ann and Jerrum, Mark},
  journal={SIAM Journal on Computing},
  volume={43},
  number={6},
  pages={1921--1952},
  year={2014},
  publisher={SIAM}
}

@article{Kuperberg2015JonesApprox,
  title={How Hard Is It to Approximate the Jones Polynomial?},
  author={Kuperberg, Greg},
  journal={Theory OF Computing},
  volume={11},
  number={6},
  pages={183--219},
  year={2015}
}

@article{childs2010Quantum,
  title = {Quantum Algorithms for Algebraic Problems},
  author = {Childs, Andrew M. and {van Dam}, Wim},
  year = 2010,
  journal = {Reviews of Modern Physics},
  volume = {82},
  number = {1},
  pages = {1--52},
  publisher = {American Physical Society}
}

@article{cleve1998Quantum,
  author  = {Cleve, Richard and Ekert, Artur and Macchiavello, Chiara and Mosca, Michele},
  title   = {Quantum Algorithms Revisited},
  journal = {Proceedings of the Royal Society of London. Series A: Mathematical, Physical and Engineering Sciences},
  volume  = {454},
  number  = {1969},
  pages   = {339--354},
  year    = {1998}
}

@article{Kasprzyk2022laurent,
 author = {Kasprzyk, Alexander and Przyjalkowski, Victor},
 title = {Laurent polynomials in mirror symmetry: why and how?},
 fjournal = {Proyecciones},
 journal = {Proyecciones},
 issn = {0716-0917},
 volume = {41},
 number = {2},
 pages = {481--515},
 year = {2022},
 language = {English},
 zbMATH = {7535440},
 Zbl = {1496.14044}
}

@inproceedings{aaronson2019quantum,
  title     = {On the Quantum Complexity of Closest Pair and Related Problems},
  author    = {Aaronson, Scott and Chia, Nai-Hui and Lin, Han-Hsuan and Wang, Chunhao and Zhang, Ruizhe},
  booktitle = {35th Computational Complexity Conference ({CCC} 2020)},
  series    = {LIPIcs},
  volume    = {169},
  pages     = {16:1--16:37},
  year      = {2020},
  publisher = {Schloss Dagstuhl--Leibniz-Zentrum f{\"u}r Informatik}
}

@article{agrawal2004primes,
  title={{PRIMES} is in {P}},
  author={Agrawal, Manindra and Kayal, Neeraj and Saxena, Nitin},
  journal={Annals of mathematics},
  pages={781--793},
  year={2004},
  publisher={JSTOR}
}

@incollection{karp2009reducibility,
  title={Reducibility among combinatorial problems},
  author={Karp, Richard M},
  booktitle={50 Years of Integer Programming 1958-2008: from the Early Years to the State-of-the-Art},
  pages={219--241},
  year={2009},
  publisher={Springer}
}

\end{document}